\documentclass[draftcls, one column]{IEEEtran}
\usepackage{amsmath,amssymb}
\usepackage{graphicx}
\usepackage{mathrsfs}
\usepackage{tikz}
\usepackage{booktabs}
\usepackage{enumerate}
\usepackage{psfrag}	
\usepackage{array}
\usepackage{multirow,hhline}
\usepackage{exscale}
\usepackage{color}
\usepackage{colortbl}
\usepackage{epsfig,subfigure}
\usepackage{cite}
\usepackage{amsthm}
\usepackage{hyperref}
\usepackage{dsfont}
\newcommand{\mat}[1]{\ensuremath{\mathbf{#1}}}
\newcommand{\alert}[1]{\textcolor{black}{#1}}
\renewcommand{\vec}[1]{\ensuremath{\mathbf{#1}}}
\newtheorem{theorem}{Theorem}
\newtheorem{lemma}{Lemma}
\newtheorem{remark}{Remark}
\newtheorem{corollary}{Corollary}

\newtheorem{definition}{Definition}

\renewcommand{\sb}[0]{\ensuremath{\mathsf{b}}}
\renewcommand{\sc}[0]{\ensuremath{\mathsf{c}}}
\newcommand{\su}[0]{\ensuremath{\mathsf{u}}}

\begin{document}

\title{The Approximate Capacity Region of the Gaussian Y-Channel}
\author{Anas~Chaaban,~\IEEEmembership{Student Member,~IEEE}, Aydin~Sezgin,~\IEEEmembership{Senior Member,~IEEE}
\thanks{The authors are with the Chair of Digital Communication Systems, Ruhr-Universit\"at Bochum (RUB), Universit\"atsstrasse 150, 44780 Bochum, Germany. Email: anas.chaaban@rub.de, aydin.sezgin@rub.de.}

\thanks{This work is supported by the German Research Foundation, Deutsche
Forschungsgemeinschaft (DFG), Germany, under grant SE 1697/3.} 

\thanks{Part of the paper has been presented in ISIT 2012 \cite{ChaabanSezginISIT} and 2013 \cite{ChaabanSezgin_ISIT12_Y}.}

}

\maketitle

\begin{abstract}
A full-duplex wireless network with three users that want to establish full message-exchange via a relay is considered. Thus, the network known as the Y-channel has a total of 6 messages, 2 outgoing and 2 incoming at each user. The users are not physically connected, and thus the relay is essential for their communication. The linear-shift deterministic Y-channel is considered first, its capacity region is characterized and shown not to be given by the cut-set bounds. The capacity achieving scheme has three different components (strategies): a bi-directional, a cyclic, and a uni-directional strategy. Network coding is used to realize the bi-directional and the cyclic strategies, and thus to prove the achievability of the capacity region. The result is then extended to the Gaussian Y-channel where the capacity region is characterized within a constant gap independent of the channel parameters.
\end{abstract}


\section{Introduction}
\IEEEPARstart{M}{ulti}-way communication refers to scenarios where nodes communicate with each other in a bi-directional manner. That is, nodes can be sources and destinations at the same time. The first studied multi-way communications setup is the two-way channel \cite{Shannon_TWC} where 2 nodes communicate with each other, and each has a message to deliver to the other node. The capacity of this setup is not known in general.

Several extensions of this setup were studied in the past decade, with more nodes and different message exchange scenarios. One such extension is obtained by combining relaying and multi-way communications to obtain the so-called multi-way relay channel. For instance, in the two-way relay channel (or the bi-directional relay channel (BRC)), two nodes communicate with each other via a relay. This extension models scenarios where communicating nodes are distributed and are only connected by intermediate nodes such as satellites. The BRC was introduced in \cite{RankovWittneben}, where relaying protocols were analyzed. Further research in this direction include work in \cite{KimDevroyeMitranTarokh} where transmission schemes employing several classical strategies (amplify, decode, and compress-forward) have been compared. In \cite{GunduzTuncelNayak, AvestimehrSezginTse, WilsonNarayananPfisterSprintson}, an approximate characterization of the capacity region of the Gaussian BRC was given. \alert{Furthermore, bi-directional communication with more users has been studied in \cite{SezginAvestimehrKhajehnejadHassibi} where the approximate capacity region of the multi-pair BRC has been characterized.}

The multi-way relay channel (MRC), consisting of more than 2 users and a relay, was studied in \cite{GunduzYenerGoldsmithPoor} where in this case, users communicate in a multi-way manner by multi-casting a message to other users via the relay. Upper and lower bounds for the capacity of the Gaussian MRC were given. In their setup, G\"und\"uz {\it et al.} divided users into several clusters, where each user in a cluster has a single message intended to all other users in the same cluster, which is referred to as multi-cast. On the other hand, Ong {\it et al.} considered a similar setup \cite{OngKellettJohnson}, where all users belong to the same cluster and all channel gains are equal. The authors of \cite{OngKellettJohnson} obtained the sum-capacity of this Gaussian setup with more than 2 users. 

A broadcast variant of this multi-way relaying setup, the so called Y-channel, was considered in \cite{LeeLimChun} where the nodes have multiple antennas. That is, three multiple antenna nodes communicate via a multiple antenna relay, and each node has two messages to broadcast to the other nodes. Each node in the Y-channel is thus a source of 2 messages and a destination of 2 messages. A transmission scheme exploiting interference alignment \cite{MaddahAliMotahariKhandani_XChannel,CadambeJafar_KUserIC} was proposed, and its corresponding achievable degrees of freedom were calculated. Note that the capacity of the Y-channel is not known in general. However, in \cite{LeeLimChun}, it was shown that if the relay has more than $\lceil{3M/2}\rceil$ antennas where $M$ is the number of antennas at the other nodes, then the sum-capacity cut-set bound \cite{CoverThomas} is asymptotically achievable, thus characterizing the degrees of freedom (DoF) of the MIMO Y-channel under this condition.

\alert{Notice that the work in \cite{LeeLimChun} considered a special case of the MIMO Y-channel, and thus the problem of the tightness of the cut-set bound in the general Gaussian Y-channel is left open. In order to resolve this problem, \cite{ChaabanSezginAvestimehr_YC_SC} has studied the SISO Y-channel, where all nodes have single antennas, and shown that the cut-set bounds are not tight, not even in the asymptotic sense (DoF sense). Note that the SISO Y-channel does not fall under the special case considered in \cite{LeeLimChun} and hence, the statement in \cite{LeeLimChun} does not apply here. The results of \cite{ChaabanSezginAvestimehr_YC_SC} showed that new bounds are required for an approximate characterization of the sum-capacity of the SISO Y-channel within a constant gap. The ideas in \cite{ChaabanSezginAvestimehr_YC_SC} have lead later on to \cite{ChaabanOchsSezgin}, where the DoF characterization of the MIMO Y-channel with an arbitrary number of antennas at all nodes has been completed. The work on the MIMO Y-channel has been pushed further to include cases with more users. For instance, \cite{LeeLeeLee} derived achievable DoF for the $K$-user MIMO Y-channel under some conditions on the numbers of antennas, and \cite{TianYenerMIMOMW} provided DoF characterization for some MIMO multi-way relay channels with $L$ clusters of users and $K$ users per cluster.}

\alert{As an approximate sum-capacity characterization for the SISO Y-channel was provided in \cite{ChaabanSezginAvestimehr_YC_SC}, the next goal is to determine the capacity region of the network within a constant gap. This is the challenge we take in this paper. It turns out that the bounds provided in \cite{ChaabanSezginAvestimehr_YC_SC} suffice for an approximate capacity region characterization of the Gaussian SISO Y-channel (GYC) within a constant gap, where the achievability is shown by using novel transmission strategies as described next.}

While interference alignment in space was used in the multi-antenna case \cite{LeeLimChun}, we use a different kind of alignment. Namely, in our single antenna case, we use alignment in the utilized codes to obtain a finite gap characterization of the capacity region of the GYC. This is accomplished by using lattice codes \alert{\cite{ErezZamir} that are aligned in such a way that facilitates computation at the relay \cite{NazerGastpar}.}

\alert{But before we study the GYC, we gain insights from the simpler linear-shift deterministic Y-channel (DYC), which is easier to handle. This special deterministic approximation has been proposed as a tool for approximating the capacity of wireless networks by Avestimehr {\it et al.} \cite{AvestimehrDiggaviTse}. By obtaining the capacity of the deterministic approximation of some wireless network, we can draw conclusions on its capacity in the Gaussian variant. For instance, the capacity region of the deterministic BRC was obtained in \cite{AvestimehrSezginTse} and used to obtain the capacity region of the Gaussian BRC relay channel within a constant gap. Similarly, the capacity of the deterministic multi-pair BRC was derived in \cite{AvestimehrKhajehnejadSezginHassibi}, which lead to the approximate capacity of the Gaussian counterpart in \cite{SezginKhajehnejadAvestimehrHassibi}.}

Our contribution for the deterministic Y-channel can be summarized as follows:
\begin{itemize}
\item We provide a new outer bound on the capacity region which is tighter than the cut-set outer bound. This is contrary to the deterministic BRC \cite{AvestimehrSezginTse} and the multi-pair BRC \cite{AvestimehrKhajehnejadSezginHassibi} where the cut-set bounds characterize the capacity region.
\item We then show that our outer bound is achievable. Network coding is used to achieve this capacity region, in a scheme which combines bi-directional, uni-directional, and a novel cyclic communication strategy.
\item Consequently, we characterize the capacity region of the DYC\footnote{\alert{It is worth to mention that the capacity of the 4-user deterministic case has been characterized recently in \cite{ZewailMohassebNafieElGamal}}.}.
\end{itemize}

Then, we use the capacity achieving scheme of the DYC to build an achievable scheme for the GYC. Superposition coding, nested lattices, and successive decoding are the main components of the proposed scheme in the Gaussian case. We provide an achievable rate region using this scheme, compare it to an outer bound, and show that it characterizes the 6-dimensional capacity region of this setup within a constant gap of 7/6 bits per dimension, regardless of the channel parameters.

The rest on the paper is organized as follows. The used notation and the system model are given in section \ref{Sec:NotationAndSystemModel}. Upper bounds for the deterministic setup are given in section \ref{Sec:DYCUpperBounds} and the capacity achieving transmit strategy is described in section \ref{Sec:DYCAchievability}. The Gaussian Y-channel is considered next, with an outer bound in section \ref{Sec:GYCOuterBound} and an inner bound in section \ref{Sec:GYCInnerBound}. The gap between the outer and the inner bounds is analyzed in section \ref{Sec:Gap}. We conclude the paper with a discussion in section \ref{Sec:Discussion}.

\section{Notation and System Model}
\label{Sec:NotationAndSystemModel}

\subsection{Notation}
Throughout the paper, we use the following notation. Scalars are represented by normal font, vectors and matrices by bold face font, and sets by calligraphic font. For instance, $x$, $\vec{x}$, and $\mathcal{X}$ are a scalar, a vector, and a set respectively. The set $\mathcal{S}^\sc$ denotes the complement of a set $\mathcal{S}$. A sequence of $n$-vectors $(\vec{x}_1,\dots,\vec{x}_n)$ is denoted $\vec{x}^n$. The modulo-2 addition (XOR) of symbols in the binary field $\mathbb{F}_2$ is denoted $\oplus$. The function $C(x)$ is given by $\frac{1}{2}\log(1+x)$ and $C^+(x)$ denotes $\max\{0,C(x)\}$.

\subsection{System Model}
The Y-channel is a multi-way relaying setup where 3 users communicate with each other in a bi-directional manner via a relay. That is, each user has a message to each other user, resulting in two outgoing messages and two incoming messages at each user, for a total of 6 messages. The users do not have direct links between each other, and hence the relay is essential for communication. This setup is shown in Figure \ref{Fig:Model}.

\begin{figure}[t]
\centering
\includegraphics[width=0.8\columnwidth]{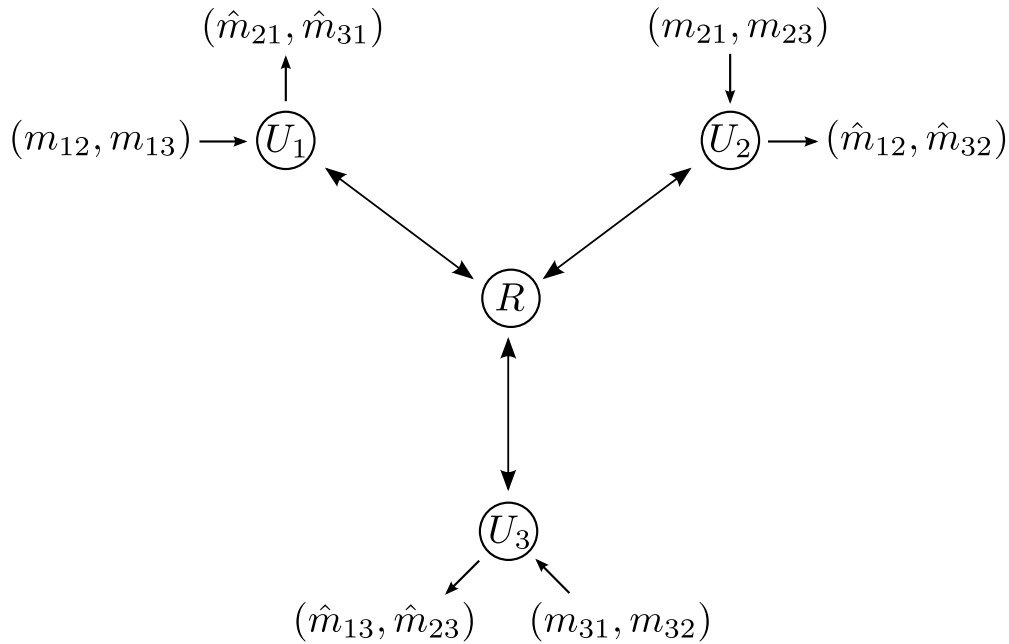}
\caption{The Y-channel showing incoming messages and outgoing messages at each node.}
\label{Fig:Model}
\end{figure}

As aforementioned, we have 6 messages, where the message $m_{jk}$ from user $j$ to user $k$, $j,k\in\{1,2,3\}$ is a \alert{realization of a random variable $M_{jk}$} uniformly distributed over the set $\mathcal{M}_{jk}\triangleq\{1,\dots,2^{nR_{jk}}\}$ for all $j\neq k$ where $R_{jk}\in\mathbb{R}_+$. For communication, each user $j$ sends a codeword of length $n$, $x_j^n$, with symbols from an alphabet $\mathcal{X}_j$. The relay receives the length-$n$ signal denoted $y_r^n\in\mathcal{Y}_r^n$, which is then processed and forwarded as the relay transmit signal $x_r^n\in\mathcal{X}_r^n$. This signal is received then at the users as $y_j^n\in\mathcal{Y}_j^n$. The sets $\mathcal{Y}_r$, $\mathcal{Y}_j$, and $\mathcal{X}_r$ are the alphabets of the relay received signal, the users' received signal, and the relay transmit signal, respectively.

The $i$th symbol of $x_j^n$ is in general a function of the outgoing messages at user $j$ and its received symbols up to time instant $i$, i.e., 
\begin{align}
\label{UsersEncoder}
x_{ji}=f_{ji}(m_{jk},m_{jl},y_j^{i-1})
\end{align}
where $f_{ji}$ is the encoding function of user $j$ at time instant $i$ with
$$f_{ji}: \mathcal{M}_{jk}\times\mathcal{M}_{jl}\times\mathcal{Y}_j^{i-1}\to\mathcal{X}_j.$$
The relay listens to the transmission of the users, constructs the signal $x_r^n$ whose $i$th symbol is 
\begin{align}
\label{RelayEncoder}
x_{ri}=f_{ri}(y_r^{i-1})
\end{align}
and is a function $f_{ri}$ of the received symbols at the relay $y_r^{i-1}$, and sends it back to the users. User $j$ receives $y_j^n$, and tries to decode his desired messages $(\hat{m}_{kj},\hat{m}_{lj})$ from $y_j^n$ using the knowledge of $(m_{jk},m_{jl})$, i.e.,  $(\hat{m}_{kj},\hat{m}_{lj})=g_j(y_j^n,m_{jk},m_{jl})$ where $$g_j:\mathcal{Y}_j^{n}\times\mathcal{M}_{jk}\times\mathcal{M}_{jl}\to\mathcal{M}_{kj}\times\mathcal{M}_{lj}$$ is the decoding function of user $j$. An error occurs if $(\hat{m}_{kj},\hat{m}_{lj})\neq({m}_{kj},{m}_{lj})$. The collection of message sets, encoders, and decoders defines a code for the Y-channel.

\subsection{Gaussian Y-channel}
In the Gaussian Y-channel (GYC) the alphabet of the transmit and received signals is the real set $\mathbb{R}$. The relay receives
\begin{align}
y_{ri}=h_1x_{1i}+h_2x_{2i}+h_3x_{3i}+z_{ri},
\end{align}
in time instant $i$, where $z_{ri}$ is a realization of an i.i.d. Gaussian noise $Z_r\sim\mathcal{N}(0,1)$ and $h_1,h_2,h_3\in\mathbb{R}$ are the channel coefficients from the users to the relay. Without loss of generality, we assume that
\begin{align}
\label{Ordering}
h_1^2\geq h_2^2\geq h_3^2.
\end{align}
The relay transmit signal $x_r$ is received at user $j$, corrupted by noise, as
\begin{align}
\label{ReceivedSignal}
y_{ji}=h_jx_{ri}+z_{ji},
\end{align}
where $z_{ji}$ is a realization of an i.i.d. Gaussian noise $Z_j\sim\mathcal{N}(0,1)$. \alert{All nodes have a power constraint $P$, thus $\frac{1}{n}\sum_{i=1}^n\mathbb{E}[X_{ri}^2]\leq P$ and 
$\frac{1}{n}\sum_{i=1}^n\mathbb{E}[X_{ji}^2]\leq P$. Note that we have assumed that the channels are reciprocal, i.e., the channel coefficient from user $j$ to the relay is the same as that from the relay to user $j$.}

\subsection{Linear-shift Deterministic Y-channel}
In the linear-shift deterministic Y-channel (DYC), the channel gains of the Gaussian Y-channel are modeled by non-negative integers \alert{$n_j=\left\lceil\frac{1}{2}\log(h_j^2P)\right\rceil$,} $j\in\{1,2,3\}$ (see \cite{AvestimehrDiggaviTse} for more details). Due to \eqref{Ordering}, we have
\begin{align}
\label{D-Ordering}
n_1\geq n_2\geq n_3.
\end{align}
These integers $n_j$, referred to as levels, define the number of bits that survive a channel which clips a number of bits of the transmitted binary vector. In more detail, the transmit signal of user $j$ and the relay is a $q$-dimensional binary vector $\vec{x}_{ji},\vec{x}_{ri}\in\mathbb{F}_2^q$ where $q=\max_j\{n_j\}$. The received signal at each node, $\vec{y}_{ri},\vec{y}_{ji}\in\mathbb{F}_2^q$, is a deterministic function of the transmit signals, modeled by a down-shift of the transmit signal. That is
\begin{align}
\label{IO1}
\vec{y}_{ri}&=\sum_{j=1}^3 \mat{S}^{q-n_j}\vec{x}_{ji}\\
\label{IO2}
\vec{y}_{ji}&=\mat{S}^{q-n_j}\vec{x}_{ri}
\end{align}
where $\vec{S}$ is the $q\times q$ downward shift matrix. Note that the impact of this channel is clipping the least significant bits of the channel input, leaving the most significant bits `visible' at the receivers. The weaker the channel (small $n_j$) the more symbols will be lost through the channel.

All operations are performed in $\mathbb{F}_2$. A deterministic Y-channel with levels $n_1$, $n_2$ and $n_3$ is denoted DYC$(n_1,n_2,n_3)$. As an example, a DYC$(4,3,2)$ is shown in Figure \ref{Fig:DYC_Model}. A line between two circles in Figure \ref{Fig:DYC_Model} represents a bit-pipe between these two levels, which models \eqref{IO1} and \eqref{IO2}.

\begin{definition}
A rate tuple $(R_{12},R_{13},R_{21},R_{23},R_{31},R_{32})$ denoted $\vec{R}$ corresponding to the message tuple \\$(m_{12},m_{13},m_{21},m_{23},m_{31},m_{32})$ denoted $\vec{m}$, is said to be achievable if there exist a sequence of codes such that the average error probability can be made arbitrarily small by increasing $n$. The set of all achievable rate tuples is the capacity region denoted $\mathcal{C}_g$ for the GYC and $\mathcal{C}_d$ for the DYC.
\end{definition}

\begin{figure}[t]
\centering
\includegraphics[width=.95\columnwidth]{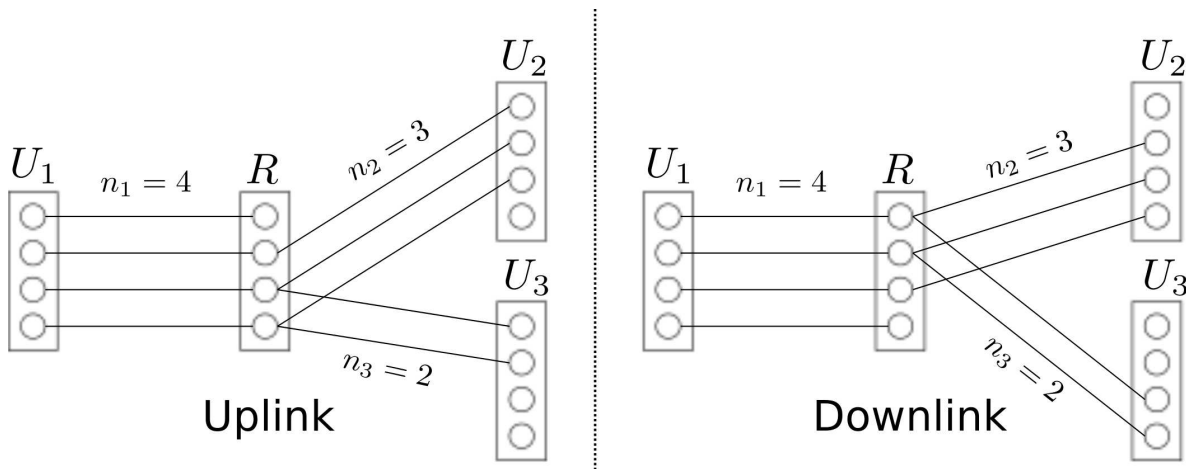}
\caption{A deterministic Y-channel with $(n_1,n_2,n_3)=(4,3,2)$. The strongest channel is between the relay ($R$) and user 1 ($U_1$). In the uplink, the relay receives all bits from user 1 that are above the noise level. Users 2 and 3 have weaker channels, and thus, bits at low levels arrive below the noise level and are clipped. Similarly in the downlink, the bits at the lower levels at the relay are clipped at receivers 2 and 3.}
\label{Fig:DYC_Model}
\end{figure}

\section{The DYC: Upper Bounds}
\label{Sec:DYCUpperBounds}
In this section, we provide some upper bounds on the achievable rates of the DYC. We start with the single rate bounds given by
\begin{align}
\label{SUB}
R_{jk}\leq\min\{n_j,n_k\},
\end{align}
followed by the cut-set bounds in the following subsection, and genie-aided upper bounds in Section \ref{SubSec:GenieAidedBounds}.

\subsection{The DYC Cut-set Bounds}
The cut-set bounds \cite{CoverThomas} can be used to obtain upper bounds on the achievable rates. Consider Figure \ref{Fig:DYC_Model}. One cut in this setup provides the sets $\mathcal{S}=\{U_1\}$ and $\mathcal{S}^\sc=\{U_2,U_3,R\}$. The rate of information flow from the set $\mathcal{S}$ into the set $\mathcal{S}^\sc$ can be bounded using the cut-set bound by
\begin{align}
R_{12}+R_{13}\leq n_1.
\end{align}
The rate in the other direction, i.e., from the set $\mathcal{S}^\sc$ into the set $\mathcal{S}$ can be bounded by
\begin{align}
R_{21}+R_{31}\leq n_1.
\end{align}
The next cut we apply gives $\mathcal{S}=\{U_1,R\}$ and $\mathcal{S}^\sc=\{U_2,U_3\}$. Here, the rate of information flow from the set $\mathcal{S}$ into the set $\mathcal{S}^\sc$ is bounded by
\begin{align}
R_{12}+R_{13}\leq \max\{n_2,n_3\}\stackrel{\eqref{D-Ordering}}{=}n_2,
\end{align}
which resembles the broadcast channel (BC) bound in \cite{AvestimehrDiggaviTse}. The rate in the other direction can be bounded by
\begin{align}
R_{21}+R_{31}\leq \max\{n_2,n_3\}\stackrel{\eqref{D-Ordering}}{=}n_2,
\end{align}
which resembles the multiple access channel (MAC) bound in \cite{AvestimehrDiggaviTse}. Following this procedure for the other remaining cuts, and collecting the resulting bounds, the cut-set bound for the DYC can be written as follows
\begin{align}
\label{CSB1}
R_{jk}+R_{jl}&\leq \min\{n_j,\max\{n_k,n_l\}\}\\
\label{CSB2}
R_{kj}+R_{lj}&\leq \min\{n_j,\max\{n_k,n_l\}\}.
\end{align}
for all distinct $j,k,l\in\{1,2,3\}$. Notice that the expressions on the left hand side of \eqref{CSB1} and \eqref{CSB2} are sums of two rates. These bounds already provide an outer bound on the capacity region $\mathcal{C}_d$. 

In many cases with bi-directional communication, the cut-set bounds were shown to characterize the whole capacity region of the deterministic setup \cite{AvestimehrKhajehnejadSezginHassibi,AvestimehrSezginTse}. However, in some deterministic bi-directional setups, the cut-set bounds are not enough for characterizing the capacity region, as in \cite{MokhtarMohassebNafieElGamal} for instance, and further bounds are required. The DYC belongs to the latter case. In fact, many rate constraints from the cut-set bounds will be shown to be redundant due to the bounds we provide next which are  more binding.

\subsection{Genie Aided Upper Bounds for the DYC}
\label{SubSec:GenieAidedBounds}
The following lemmas provide upper bounds on the achievable rates of the DYC which are tighter than some cut-set bounds. They are obtained by giving additional side information to some nodes.

\begin{lemma}
\label{FromRelay}
The achievable rates in the DYC must satisfy
\begin{align}
\label{BFR}
R_{kj}+R_{lj}+R_{kl}&\leq \max\{n_j,n_l\}
\end{align}
for all distinct $j,k,l\in\{1,2,3\}$.
\end{lemma}

\begin{proof}
Let a genie give the message $m_{32}$ to $U_1$ as additional information. This gives a new channel (a genie-aided channel) with a larger capacity region, and hence leads to a valid upper bound. \alert{After decoding $m_{21}$ and $m_{31}$, $U_1$ will have knowledge of the following information $$(\vec{Y}_1^n,m_{31},m_{32}).$$ Recall that $U_3$ is able to decode $m_{23}$ from $(\vec{Y}_3^n,m_{31},m_{32})$. Since $U_1$ has $(\vec{Y}_1^n,m_{31},m_{32})$ which is a better observation than $(\vec{Y}_3^n,m_{31},m_{32})$ (see \eqref{D-Ordering} and \eqref{IO2}), it can decode $(m_{21},m_{31},m_{23})$ leading to the rate bound}
\begin{align*}
&\hspace{-1cm}n(R_{21}+R_{31}+R_{23}-\epsilon_{n})\nonumber\\
&\leq I(m_{21},m_{31},m_{23};\vec{Y}_1^n,m_{12},m_{13},m_{32})
\end{align*}
where $\epsilon_{n}\to0$ as $n\to\infty$ from Fano's inequality. We continue as follows
\begin{align*}
n(R_{21}+R_{31}&+R_{23}-\epsilon_{n})\\
&\stackrel{(a)}{\leq} I(m_{21},m_{31},m_{23};\vec{Y}_1^n|m_{12},m_{13},m_{32})\\
&= H(\vec{Y}_1^n|m_{12},m_{13},m_{32})-H(\vec{Y}_1^n|\vec{m})\\
&\stackrel{(b)}{\leq} H(\vec{Y}_1^n)-H(\vec{Y}_1^n|\vec{m},\vec{X}_r^n)\\
&\stackrel{(c)}{=} H(\vec{Y}_1^n)\\
&\stackrel{(d)}{=} \sum_{i=1}^n H(\vec{Y}_{1i}|\vec{Y}_{1}^{i-1})\\
&\stackrel{(b)}{\leq} \sum_{i=1}^n H(\vec{Y}_{1i})\\
&\stackrel{(d)}{=} \sum_{i=1}^n \sum_{p=1}^{n_1} H(Y_{1i}(p)|Y_{1i}(1),\dots,Y_{1i}(p-1))\\
&\stackrel{(b)}{\leq} \sum_{i=1}^n \sum_{p=1}^{n_1} H(Y_{1i}(p))\\
&\stackrel{(e)}{\leq} n(n_1)
\end{align*}
where $Y_{1i}(p)$ is the $p$th component of $\vec{Y}_{1i}$, and
\begin{itemize}
\item[$(a)$] follows due to the independence of the messages,
\item[$(b)$] follows since conditioning does not increase entropy, 
\item[$(c)$] follows since $H(\vec{Y}_1^n|\vec{m},\vec{X}_r^n)=0$ because $\vec{Y}_1^n$ is a deterministic function of $\vec{X}_r^n$, 
\item[$(d)$] follows by the chain rule, and
\item[$(e)$] follows since the binary entropy function is maximized to 1 by the Bernoulli distribution with probability $0.5$. 
\end{itemize}
Thus, with $n\to\infty$, 
\begin{align*}
R_{21}+R_{31}+R_{23}\leq n_1.
\end{align*}
In a similar way, we can obtain the other bounds and the lemma is proved.
\end{proof}

\begin{lemma}
\label{GUB}
The achievable rates in the DYC must satisfy
\begin{align}
\label{BTR}
R_{kj}+R_{lj}+R_{kl}&\leq \max\{n_k,n_l\}
\end{align}
for all distinct $j,k,l\in\{1,2,3\}$.
\end{lemma}

\begin{proof}
\alert{A genie gives $(\vec{Y}_r^n,m_{32})$ as side information to the $U_1$. Then, similar to the proof of Lemma \ref{FromRelay}, $U_1$ can decode $m_{23}$ after decoding $m_{21}$ and $m_{31}$. In more detail, after decoding $m_{21}$ and $m_{31}$, $U_1$ has the observation $(\vec{Y}_r^n,m_{31},m_{32})$. This observation is a better observation than that of $U_3$. Hence, $U_1$ can also decode $m_{23}$. Now we can use Fano's inequality to write the bound}
\begin{align}
n(R_{21}&+R_{31}+R_{23}-\epsilon_{n})\nonumber\\
&\leq I(m_{21},m_{31},m_{23};\vec{Y}_1^n,\vec{Y}_r^n,m_{12},m_{13},m_{32}),
\end{align}
where $\epsilon_{n}\to 0$ as $n\to\infty$. We proceed with this bound as follows
\begin{align*}
n(R_{21}&+R_{31}+R_{23}-\epsilon_{n})\nonumber\\
&\stackrel{(a)}{\leq} I(m_{21},m_{31},m_{23};\vec{Y}_1^n,\vec{Y}_r^n|m_{12},m_{13},m_{32})\nonumber\\
&\stackrel{(b)}{=} I(m_{21},m_{31},m_{23};\vec{Y}_r^n|m_{12},m_{13},m_{32})\nonumber\\
&= H(\vec{Y}_r^n|m_{12},m_{13},m_{32})-H(\vec{Y}_r^n|\vec{m})\nonumber\\
&\leq H(\vec{Y}_r^n|m_{12},m_{13},m_{32})\nonumber\\
&\stackrel{(c)}{=} \sum_{i=1}^n H(\vec{Y}_{ri}|m_{12},m_{13},m_{32},\vec{Y}_r^{i-1})\nonumber\\
&\stackrel{(d)}{=} \sum_{i=1}^n H(\vec{Y}_{ri}|m_{12},m_{13},m_{32},\vec{Y}_r^{i-1},\vec{Y}_1^{i},\vec{X}_1^{i+1})\nonumber\\
&\stackrel{(e)}{\leq} \sum_{i=1}^n H(\vec{Y}_{ri}|\vec{X}_{1i})\nonumber\\
&\stackrel{(c)}{=} \sum_{i=1}^n \sum_{p=1}^{n_1}H(Y_{ri}(p)|Y_{ri}(1),\dots,Y_{ri}(p-1),\vec{X}_{1i})\nonumber\\
&\stackrel{(e)}{\leq} \sum_{i=1}^n \sum_{p=1}^{n_1}H(Y_{ri}(p)|\vec{X}_{1i})\nonumber\\
&\stackrel{(f)}{\leq} \sum_{i=1}^n \sum_{p=1}^{n_2}H(Y_{ri}(p))\nonumber\\
&\stackrel{(g)}{\leq} n(n_2)
\end{align*}
where where $Y_{ri}(p)$ is the $p$th component of $\vec{Y}_{ri}$, and
\begin{itemize}
\item[$(a)$] follows due to the independence of the messages, 
\item[$(b)$] follows from the Markov chain $$(m_{21},m_{31},m_{23})\to(\vec{Y}_r^n,m_{12},m_{13},m_{32})\to\vec{Y}_1^n,$$
\item[$(c)$] follows from the chain rule, 
\item[$(d)$] follows since knowing $\vec{Y}_r^{i-1}$, we can construct $\vec{X}_{ri}=f_{ri}(\vec{Y}_r^{i-1})$ for $i\in\{1,\dots,i\}$ \eqref{RelayEncoder}, i.e., we can construct $\vec{X}_r^{i}$, then we can construct $\vec{Y}_1^i$ which is a deterministic function of $\vec{X}_r^{i}$ \eqref{IO2}, and then we can use $(m_{12},m_{13},\vec{Y}_1^{i})$ to construct $\vec{X}_1^{i+1}$ since  $\vec{X}_{1i}=f_{1i}(m_{12},m_{13},\vec{Y}_1^{i-1})$ \eqref{UsersEncoder},
\item[$(e)$] follows since conditioning does not increase entropy,
\item[$(f)$] follows since knowing $\vec{X}_{1i}$, we know the value of the most significant $n_1-n_2$ bits of $\vec{Y}_{ri}$ (where no interference occurs). Hence the remaining uncertainty is that of the remaining $n_2$ bits.
\item[$(g)$] follows since the binary entropy function is maximized to 1 by the Bernoulli distribution with probability $0.5$. 
\end{itemize}
Letting $n\to\infty$, we obtain
\begin{align*}
R_{21}+R_{31}+R_{23}\leq n_2.
\end{align*}
The other bounds can be obtained in a similar way, and this concludes the proof.
\end{proof}

Notice that the bounds in Lemmas \ref{FromRelay} and \ref{GUB} constrain the sum of three components of $\vec{R}$ to be lower than a specific value $n_j$, contrary to the cut-set bounds that constrain the sum of two components of $\vec{R}$. This makes these bounds tighter than the cut-set bounds in general as we shall see next. Now, we combine Lemmas \ref{FromRelay} and \ref{GUB} to obtain the following statement.
\begin{theorem}
The achievable rates in the DYC are upper bounded by
\begin{align}
\label{TRB1}
R_{12}+R_{32}+R_{13}&\leq n_2\\
\label{TRB2}
R_{12}+R_{32}+R_{31}&\leq n_1\\
\label{TRB3}
R_{21}+R_{31}+R_{32}&\leq n_2\\
\label{TRB4}
R_{21}+R_{31}+R_{23}&\leq n_2\\
\label{TRB5}
R_{13}+R_{23}+R_{12}&\leq n_2\\
\label{TRB6}
R_{13}+R_{23}+R_{21}&\leq n_1.
\end{align}
\end{theorem}
\begin{proof}
By combining the genie-aided bounds in \eqref{BFR} and \eqref{BTR} and evaluating for all distinct $j,k,l\in\{1,2,3\}$ we obtain the statement of the theorem.
\end{proof}

Let us now evaluate the individual rate bounds \eqref{SUB}, and the cut-set bounds \eqref{CSB1} and \eqref{CSB2} using \eqref{Ordering}. The individual rate bounds become
\begin{align*}
\begin{array}{ccc}
R_{12}\leq n_2, & R_{13}\leq n_3, & R_{23}\leq n_3,\\
R_{21}\leq n_2, & R_{31}\leq n_3, & R_{32}\leq n_3.
\end{array}
\end{align*}
The cut-set bounds yield
\begin{align}
\label{CS1}
R_{12}+R_{13}\leq n_2, & \quad  R_{21}+R_{31}\leq n_2,\\
\label{CS2}
R_{21}+R_{23}\leq n_2, & \quad  R_{32}+R_{12}\leq n_2,\\
\label{CS3}
R_{31}+R_{32}\leq n_3, & \quad  R_{13}+R_{23}\leq n_3.
\end{align}
Notice that the individual rate bounds are all redundant given the cut-set bounds. For example, $R_{12}\leq n_2$ is redundant given $R_{12}+R_{13}\leq n_2$ since all rates are positive. Moreover, the cut-set bounds in \eqref{CS1} are redundant given the genie-aided bounds \eqref{TRB1} and \eqref{TRB3}. Similarly, the cut-set bounds in \eqref{CS2} are redundant given the genie-aided bounds \eqref{TRB4} and \eqref{TRB1}. Only cut-set bounds in \eqref{CS3} remain useful. 

As a result, by defining $\overline{\mathcal{C}}_d$ to be the region in $\mathbb{R}_+^6$ satisfying the genie-aided bounds \eqref{TRB1}-\eqref{TRB6} and the cut-set bounds \eqref{CS3}, that is, 
\begin{align}
\label{OBE}
\overline{\mathcal{C}}_d\triangleq\left\{
\begin{array}{rl}
\vec{R}\in\mathbb{R}_+^6:&R_{31}+R_{32}\leq n_3\\
&R_{13}+R_{23}\leq n_3\\
&R_{12}+R_{32}+R_{13}\leq n_2\\
&R_{12}+R_{32}+R_{31}\leq n_1\\
&R_{21}+R_{31}+R_{32}\leq n_2\\
&R_{21}+R_{31}+R_{23}\leq n_2\\
&R_{13}+R_{23}+R_{12}\leq n_2\\
&R_{13}+R_{23}+R_{21}\leq n_1
\end{array}
\right\},
\end{align}
we obtain the following outer bound on $\mathcal{C}_d$.
\begin{theorem}
\label{OB}
The capacity region $\mathcal{C}_d$ of the DYC is outer bounded by $\overline{\mathcal{C}}_d$.
\end{theorem}

In the next section, we show that this outer bound is achievable, and hence, we characterize the capacity region $\mathcal{C}_d$ of the DYC.

\section{A Capacity Achieving Scheme for the DYC}
\label{Sec:DYCAchievability}
We start by showing that any integer rate tuple in $\overline{\mathcal{C}}_d$ is achievable. That is, every tuple $\vec{R}\in\mathbb{N}^6\cap\overline{\mathcal{C}}_d$ is achievable. Consider any such tuple $\vec{R}$. Since $\vec{R}\in\overline{\mathcal{C}}_d$, then it satisfies the bounds in Theorem \ref{OB}. Now, we have to show that we can use the signal levels at the relay wisely to achieve this rate tuple. Our scheme uses three different strategies to cover three different modes of information flow. These modes are as follows:
\begin{itemize}
\item[$\sb$)] \textbf{Bi-directional:} There exist users that want to establish bi-directional communication. That is, $R_{jk}$ and $R_{kj}$ are both non-zero for some $j,k\in\{1,2,3\}$, $j\neq k$. 
\item[$\sc$)] \textbf{Cyclic:} Users want to establish cyclic communication. That is, $R_{jk}$, $R_{kl}$, and $R_{lj}$ are non-zero while $R_{kj}$, $R_{lk}$, and $R_{jl}$ are all zero for some distinct $j,k,l\in\{1,2,3\}$.
\item[$\su$)]\textbf{Uni-directional:} Neither case $\sb$) nor $\sc$) holds. That is, at least three components of $\vec{R}$ are zero, and the non-zero components are uni-directional (if $R_{jk}\neq0$ then $R_{kj}=0$) and acyclic (if $R_{jk},R_{kl}\neq0$ then $R_{lj}=0$).
\end{itemize}

We used $\sb$), $\sc$), and $\su$) to refer to the $\sb$i-directional, $\sc$yclic, and $\su$ni-directional modes of information flow, respectively. These three modes are taken care of in the construction of the communication strategy in the given order. That is, we design a strategy for the bi-directional mode first, then a strategy for the cyclic mode, and finally a strategy for the remaining uni-directional mode. A brief description of the scheme is given in the following toy example, more details to follow up next.

\subsection{DYC: A Toy Example}
Consider a DYC$(5,4,3)$ and choose $\vec{R}=(0,2,2,1,0,2)$. By inserting the values of $n_1$, $n_2$ and $n_3$ in the outer bound $\overline{\mathcal{C}}_d$, it is easy to see that the rate tuple $\vec{R}\in\overline{\mathcal{C}}_d$. Let us see how our scheme works for achieving this rate tuple.

We start by writing $\vec{R}$ as $$\vec{R}=\vec{R}^\sb+\vec{R}^\sc+\vec{R}^\su,$$
where $\vec{R}^\sb=(0,0,0,1,0,1)$, $\vec{R}^\sc=(0,1,1,0,0,1)$, and $\vec{R}^\su=(0,1,1,0,0,0)$. Notice that $\vec{R}^\sb$ resembles bi-directional information flow between $U_2$ and $U_3$ with a rate of 1 bit per channel use in each direction. To achieve this rate tuple, let $U_2$ send one bit $b_{23}$ on the lowest level in the uplink, i.e., relay level 1, and let $U_3$ also send 1 bit $b_{32}$ on relay level 1. Thus, the relay receives $b_{23}\oplus b_{32}$ on level 1 as shown in Figure \ref{DYC543U}.

\begin{figure}[t]
\centering
\includegraphics[width=.45\columnwidth]{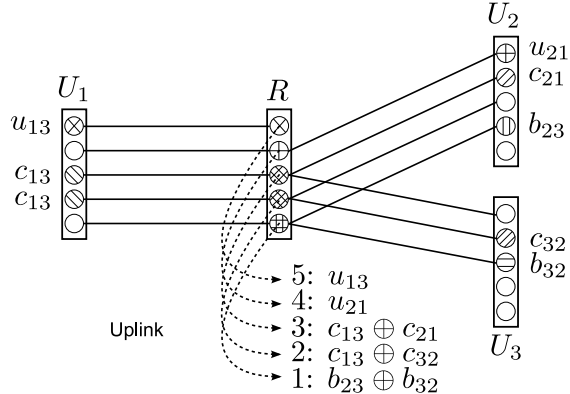}
\caption{A DYC$(5,4,3)$ with an illustration of our transmit strategy in the uplink.}
\label{DYC543U}
\end{figure}

The relay then forwards $b_{23}\oplus b_{32}$ on the highest level in the downlink, i.e.,  level 1 as shown in Figure \ref{DYC543D} (the enumeration of levels at the relay is as given in the figure). Upon receiving $b_{23}\oplus b_{32}$, $U_2$ and $U_3$ are able to extract their desired bits, $b_{32}$ and $b_{23}$, respectively. This achieves $\vec{R}^\sb$. Note that the used levels are not available anymore for further communication. This allows us to remove the used levels, to obtain a DYC$(4,3,2)$, over which we need to achieve $\vec{R}^\sc$ and $\vec{R}^\su$. 

The rate tuple $\vec{R}^\sc$ represents the rates of the cyclic information flow, where $U_1$ wants to send 1 bit $c_{13}$ to $U_3$, $U_3$ wants to send 1 bit $c_{32}$ to $U_2$, and $U_2$ wants to send 1 bit $c_{21}$ to $U_1$, thus forming the cycle $1\to3\to2\to1$. To achieve $\vec{R}^\sc$, we use a cyclic strategy. $U_1$ sends $c_{13}$ on both relay levels 2 and 3, $U_2$ sends $c_{21}$ on relay level 3, and $U_3$ sends $c_{32}$ on relay level 2. The relay thus receives $c_{13}\oplus c_{32}$ and $c_{13}\oplus c_{21}$ on levels 2 and 3, respectively, as shown in Figure \ref{DYC543U}. The relay then forwards these sums on levels 3 and 4 as shown in Figure \ref{DYC543D}. Upon receiving $c_{13}\oplus c_{21}$, $U_1$ can extract its desired bit $c_{21}$, and upon receiving $c_{13}\oplus c_{32}$, $U_3$ can extract its desired bit $c_{13}$. $U_2$ extracts $c_{13}$ from $c_{13}\oplus c_{21}$, and then uses it to extract its desired bit $c_{32}$ from $c_{13}\oplus c_{32}$. Thus, this achieves $\vec{R}^\sc$.

\begin{figure*}[t]
\centering
\includegraphics[width=.7\columnwidth]{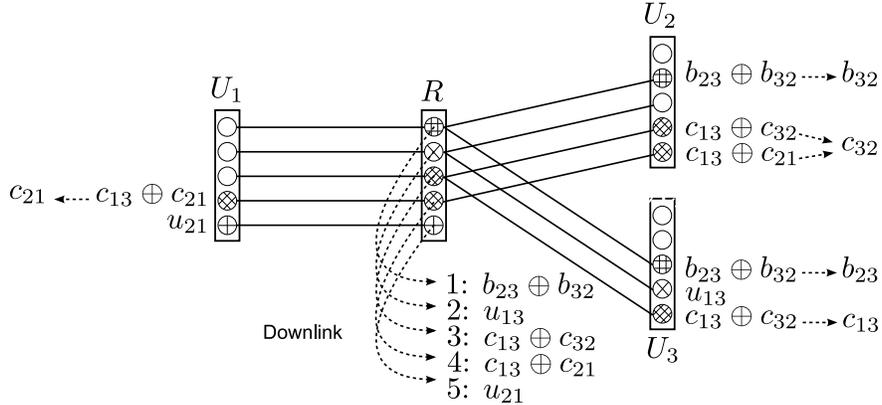}
\caption{A DYC$(5,4,3)$ with an illustration of our transmit strategy in the downlink.}
\label{DYC543D}
\end{figure*}

Finally, it remains to achieve $\vec{R}^\su$. To do this, $U_1$ and $U_2$ send one bit each, $u_{13}$ and $u_{21}$, to levels 5 and 4 at the relay, respectively, as shown in Figure \ref{DYC543U}. The relay forwards these bits on levels 2 and 5, respectively, and users $U_1$ and $U_3$ are then able to recover both desired bits. This achieves $\vec{R}^\su$. All three tuples, $\vec{R}^\sb$, $\vec{R}^\sc$, and $\vec{R}^\su$, are achieved, which consequently achieves the rate tuple $\vec{R}$. In conclusion, the users send
\begin{align*}
\vec{X}_1=\left[\begin{array}{c}u_{13}\\0\\c_{13}\\c_{13}\\0\end{array}\right],\quad \vec{X}_2=\left[\begin{array}{c}u_{21}\\c_{21}\\0\\b_{23}\\0\end{array}\right],\quad
\vec{X}_3=\left[\begin{array}{c}0\\c_{32}\\b_{32}\\0\\0\end{array}\right],
\end{align*}
and the relay shuffles its received signal $\vec{Y}_r$ as follows
\begin{align*}
\vec{X}_r=\left[\begin{array}{ccccc}
0&0&0&0&1\\
1&0&0&0&0\\
0&0&0&1&0\\
0&0&1&0&0\\
0&1&0&0&0
\end{array}\right]\vec{Y}_r.
\end{align*}

\begin{remark}
Notice that all the levels at the relay have been used to achieve $\vec{R}$. Here, we can see that importance of the cyclic strategy. The cyclic strategy uses 2 levels at the relay for communicating 3 bits, i.e., it sends 3/2 bits per level. If we want to achieve $\vec{R}^\sc$ using the uni-directional strategy which sends 1 bit per level at the relay, instead of the cyclic one, then we would consume 3 levels at the relay instead of 2. This leaves us with 1 more level, which is not sufficient to achieve $\vec{R}^\su$.
\end{remark}

\begin{remark}
The problem of cyclic communication bears a resemblance to an index coding problem given in \cite{BarYossefBirkJayramKol}. Consider the cycle $1\to3\to2\to1$. If  the relay knows the bits $c_{13}$, $c_{32}$, and $c_{21}$, then the downlink is similar to \cite[Example 1]{BarYossefBirkJayramKol} where the side information graph is a directed cycle of length 3. Namely, nodes $(U_1,U_2,U_3)$ requesting $(c_{21},c_{32},c_{13})$, know $(c_{13},c_{21},c_{32})$, respectively, in the given order. The optimal linear code in this case is of length 2. The relay sends $(c_{13}\oplus c_{32},c_{13}\oplus c_{21})$ which enables all receivers to recover their desired bits. In our case, the relay does not know the source bits. However, this index code can be constructed in our case on the fly, i.e., during the uplink. The uplink scheme is designed in such a way that the relay obtains $(c_{13}\oplus c_{32},c_{13}\oplus c_{21})$ which suffices for sending all 3 bits to their destinations.
\end{remark}

Next, we use these strategies to show the achievability of any integer valued $\vec{R}$. Briefly, the proof proceeds as follows. For any such $\vec{R}$, the proposed scheme starts with a bi-directional communication strategy if case $\sb$) holds. The bi-directional strategy sends two bits per relay level. That is, one signal level at the relay is consumed by two bits of bi-directional streams. After this step, some rates are already achieved and the residual rate vector is called $\vec{R}'$. We also have a reduced DYC obtained by removing the already occupied levels. It remains to achieve $\vec{R}'$ which has at least three zero components over this reduced DYC. Now, we use the cyclic strategy if case $\sc$) holds, which sends 3/2 bits per relay level. After this step, the residual rate vector, denoted $\vec{R}''$ belongs to case $\su$), and we use the uni-directional strategy to achieve it, which sends 1 bit per relay level. These strategies are explained in more detail in the next subsections.

Notice the used enumeration of levels in Figures \ref{DYC543U} and \ref{DYC543D}. We will follow this enumeration throughout the rest of the paper. In the uplink, the lowest level is level 1 and the highest is level $q=n_1$. In the downlink, the lowest level is $q=n_1$ and the highest is level 1 as shown in Figure \ref{RelayLevelsU} and \ref{RelayLevelsD}. This enumeration is used for convenience. Using this enumeration, levels $\{1\dots,n_3\}$ are the levels shared by all three users in both the uplink and downlink. Similarly, levels $\{n_3+1,\dots,n_2\}$ are shared by users 1 and 2, and all remaining levels are used exclusively by user 1. The levels at the relay will be represented as a line segment with three parts representing the sets $\{1,\dots,n_3\}$, $\{n_3+1,\dots, n_2\}$, and $\{n_2+1,\dots,n_1\}$ as shown in Figure \ref{RelayLevelsU} and \ref{RelayLevelsD}.

\begin{figure}[t]
\centering
\subfigure[Relay levels in the uplink.]{
\includegraphics[width=0.22\columnwidth]{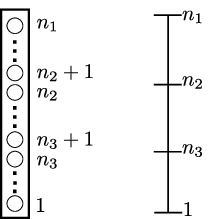}
\label{RelayLevelsU}
}
\hspace{2cm}
\subfigure[Relay levels in the downlink.]{
\includegraphics[width=0.22\columnwidth]{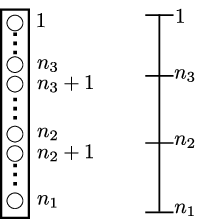}
\label{RelayLevelsD}
}
\caption{Enumeration of levels at the relay. Levels 1 to $n_3$ are seen by all three users, levels $n_3+1$ to $n_2$ are seen by users 1 and 2, while the remaining levels are seen only by user 1.}
\end{figure}

\subsection{Bi-directional information flow over the DYC}
\label{BDC}
We start by assigning levels for bi-directional information flows if case $\sb$) holds, over a DYC$(n_1,n_2,n_3)$. Let\footnote{$R_{12}^\sb$, $R_{13}^\sb$, or $R_{23}^\sb$ can have zero value. If all are zero, then the bi-directional information flow mode does not exist and we start with case $\sc$) instead.}
\begin{align}
\label{R12b}
R_{12}^\sb&=\min\{R_{12},R_{21}\},\\
\label{R13b}
R_{13}^\sb&=\min\{R_{13},R_{31}\},\\
\label{R23b}
R_{23}^\sb&=\min\{R_{23},R_{32}\}.
\end{align}
$U_1$ and $U_2$ use levels $\{n_2-R_{12}^\sb+1,\dots,n_2\}$ in a manner similar to the deterministic bi-directional relay channel \cite{AvestimehrSezginTse} to exchange $R_{12}^\sb$ bits. That is, each of users 1 and 2 sends a binary vector, say $\vec{b}_{12}$ and $\vec{b}_{21}$ where $$\vec{b}_{12},\vec{b}_{21}\in\mathbb{F}_2^{R_{12}^\sb},$$ on levels $\{n_2-R_{12}^\sb+1,\dots,n_2\}$, The relay obtains the superposition $\vec{b}_{12}\oplus\vec{b}_{21}$ and sends it back to users 1 and 2 on the same levels. Users 1 and 2 in their turn calculate their desired information from the $\vec{b}_{12}\oplus\vec{b}_{21}$ using their transmit vector as side information. Similarly, users 1 and 3 use levels $\{1,\dots,R_{13}^\sb\}$, given $R_{13}^\sb\leq n_3$ so that user 3 can send and receive all $R_{13}^\sb$ bits. Users 2 and 3 use levels $\{R_{13}^\sb+1,\dots,R_{13}^\sb+R_{23}^\sb\}$ where $R_{13}^\sb+R_{23}^\sb\leq n_3$ is required for the same reason (which is stronger than the former $R_{13}^\sb\leq n_3$). This is shown graphically in Figure \ref{RelayLevels1}.

This strategy works if we have enough levels at the relay for all $R_{12}^\sb+R_{13}^\sb+R_{23}^\sb$ bi-directional streams (for delivering twice the number of bits). Thus it is required that $R_{12}^\sb+R_{13}^\sb+R_{23}^\sb\leq n_2$ in addition to $R_{13}^\sb+R_{23}^\sb\leq n_3$. But these inequalities hold as long as $\vec{R}\in\overline{\mathcal{C}}_d$ since
\begin{align}
R_{13}^\sb+R_{23}^\sb&\leq R_{13}+R_{23}\\
&\stackrel{\eqref{CS3}}{\leq} n_3
\end{align}
and
\begin{align}
R_{12}^\sb+R_{13}^\sb+R_{23}^\sb&\leq R_{12}+R_{13}+R_{23}\\
&\stackrel{\eqref{TRB5}}{\leq} n_2.
\end{align}
It follows that the levels at the relay are sufficient for this strategy to work. Now having communicated $R_{12}^\sb+R_{13}^\sb+R_{23}^\sb$ bi-directional streams ($2(R_{12}^\sb+R_{13}^\sb+R_{23}^\sb)$ bits), the remaining rate tuple that needs to be achieved is 
\begin{align}
\label{Rp}
\vec{R}'&\triangleq(R_{12}',R_{13}',R_{21}',R_{23}',R_{31}',R_{32}'),
\end{align}
where
\begin{align}
\label{R12'}
R_{12}'&=R_{12}-R_{12}^\sb,\quad R_{13}'=R_{13}-R_{13}^\sb\\
\label{R21'}
R_{21}'&=R_{21}-R_{12}^\sb,\quad R_{23}'=R_{23}-R_{23}^\sb\\
\label{R31'}
R_{31}'&=R_{31}-R_{13}^\sb,\quad R_{32}'=R_{32}-R_{23}^\sb
\end{align}
and at least three of the components of $\vec{R}'$ are zero (cf. \eqref{R12b}, \eqref{R13b}, \eqref{R23b}). Namely, one of $R_{jk}'$ and $R_{kj}'$ must be zero. This rate tuple $\vec{R}'$ must be achieved over DYC$(n_1',n_2',n_3')$ where the already occupied levels are out of order, i.e.,
\begin{align}
\label{n1p}
n_1'=n_1-R_{12}^\sb-R_{13}^\sb-R_{23}^\sb,\\
\label{n2p}
n_2'=n_2-R_{12}^\sb-R_{13}^\sb-R_{23}^\sb.
\end{align}
If $n_2-n_3\geq R_{12}^\sb$, then the bi-directional communication between users 1 and 2 does not use levels in $\{1,\dots,n_3\}$ in which case 
\begin{align}
n_3'=n_3-R_{13}^\sb-R_{23}^\sb. 
\end{align}
Otherwise, $R_{12}^\sb-(n_2-n_3)$ levels in $\{1,\dots,n_3\}$ are used for this communication and in this case 
\begin{align}
n_3'&=n_3-R_{13}^\sb-R_{23}^\sb-(R_{12}^\sb-(n_2-n_3))\\
&=n_2'. 
\end{align}
Therefore we can write
\begin{align}
\label{n3p}
n_3'=\min\{n_3-R_{13}^\sb-R_{23}^\sb,n_2'\}.
\end{align}
The non-zero components of $\vec{R}'$ can represent cyclic information flow as in case $\sc$) or uni-directional information flow as in case $\su$) described in Section \ref{Sec:DYCAchievability}. Next we describe the cyclic case $\sc$).

\begin{figure}[t]
\centering
\includegraphics[height=0.5\columnwidth]{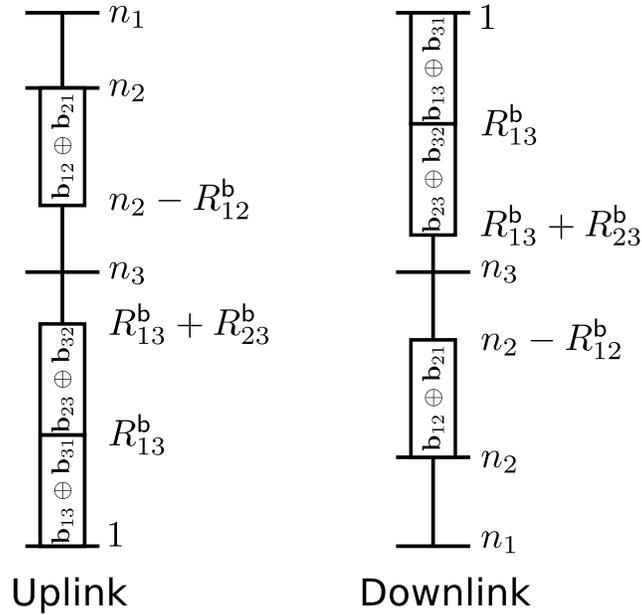}
\caption{Levels used for bi-directional communication. Note that the lower two blocks in the uplink, i.e., $2\leftrightarrow 3$ and $1\leftrightarrow 3$ can be swapped. Similarly in the downlink.}
\label{RelayLevels1}
\end{figure}

\subsection{Cyclic information flow over the DYC}
\label{CC}
After assigning levels to all bits of bi-directional information flow, we consider cyclic information flow. We need to achieve a rate tuple $\vec{R}'$ which has three zero components. If case $\sc$) holds, then users want to communicate in a cyclic manner over a DYC$(n_1',n_2',n_3')$. There are two possible cycles, either $1\to2\to3\to1$ or $1\to3\to2\to1$ (Figure \ref{Cycles}). Let\footnote{Either $R_{123}^\sc$ or $R_{132}^\sc$ must have zero value. If both are zero, then this strategy is skipped and case $\su$) is considered instead.}
\begin{align}
\label{D}
R_{123}^\sc&=\min\{R_{12}',R_{23}',R_{31}'\},\\
\label{E}
R_{132}^\sc&=\min\{R_{13}',R_{32}',R_{21}'\}.
\end{align}
Notice that
\begin{align}
&R_{123}^\sc>0\nonumber\\
\label{R1}
&\Rightarrow R_{132}^\sc=0,\ R_{12}^\sb=R_{21},\ R_{13}^\sb=R_{13},\ R_{23}^\sb=R_{32},
\end{align}
and
\begin{align}
&R_{132}^\sc>0\nonumber\\
\label{R2}
&\Rightarrow R_{123}^\sc=0,\ R_{12}^\sb=R_{12},\ R_{13}^\sb=R_{31},\ R_{23}^\sb=R_{23}.
\end{align}
which follow from \eqref{R12b}-\eqref{R23b}, \eqref{R12'}-\eqref{R31'}, \eqref{D}, and \eqref{E}. Namely, if $R_{123}^\sc>0$, then by \eqref{D}, none of $R_{12}'$, $R_{23}'$, or $R_{31}'$ is zero. Thus, by \eqref{R12'}, $R_{12}\geq R_{12}^\sb$, which by \eqref{R12b} means that $R_{12}^\sb=R_{21}$. Similarly $R_{13}^\sb=R_{13}$ and $R_{23}^\sb=R_{32}$. Then, using \eqref{E} we have $R_{132}^\sc=0$.

\begin{figure*}[t]
\centering
\subfigure[$1\to2\to3\to1$: In this case $R_{13}=R_{21}=R_{32}=0$.]{
\includegraphics[width=0.4\columnwidth]{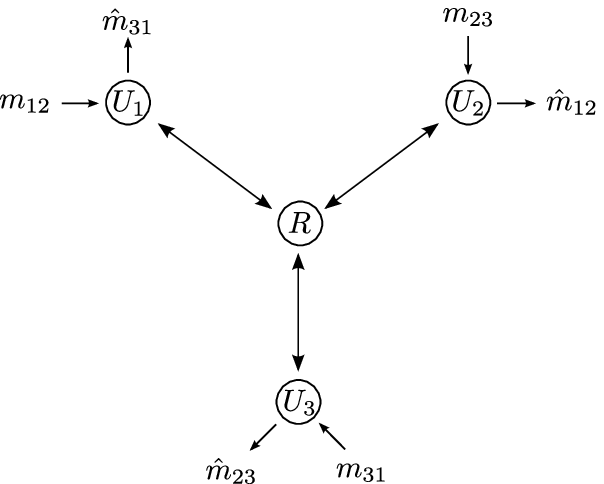}
}
\hspace{0.06\columnwidth}
\subfigure[$1\to3\to2\to1$: In this case $R_{12}=R_{31}=R_{23}=0$.]{
\includegraphics[width=0.4\columnwidth]{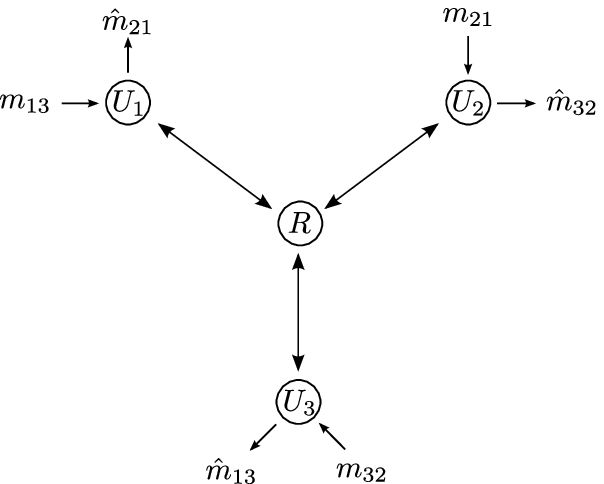}
}
\caption{Cyclic communication.}
\label{Cycles}
\end{figure*}

Let us consider the first case \eqref{R1}. Thus $R_{123}^\sc\neq 0$ and $R_{132}^\sc=0$. Let the transmit binary vectors be 
$$\vec{c}_{12},\vec{c}_{23},\vec{c}_{31}\in\mathbb{F}_2^{R_{123}^\sc}.$$

Note that we force all users to use the same rate $R_{123}^\sc$. $U_1$ and $U_2$ send their bits $\vec{c}_{12}$ and $\vec{c}_{23}$ on levels $$\{n_2'-R_{123}^\sc+1,\dots,n_2'\},$$ at the relay. $U_2$ repeats $\vec{c}_{23}$ on levels $$\{1,\dots,R_{123}^\sc\}$$ at the relay, which are also used by $U_3$ to send $\vec{c}_{31}$ (assuming that the levels are sufficient, i.e., the sets $\{n_2'-R_{123}^\sc+1,\dots,n_2'\}$ and $\{1,\dots,R_{123}^\sc\}$ do not intersect, or equivalently $2R_{123}^\sc\leq n_2'$ which we prove next).

The relay receives $\vec{c}_{12}\oplus\vec{c}_{23}$ and $\vec{c}_{23}\oplus\vec{c}_{31}$ and sends them back on the same levels at the relay ($\{n_2'-R_{123}^\sc+1,\dots,n_2'\}$ and $\{1,\dots,R_{123}^\sc\}$) in the downlink. $U_1$ and $U_2$ receive $\vec{c}_{12}\oplus\vec{c}_{23}$ and $\vec{c}_{23}\oplus\vec{c}_{31}$ since all bits are sent on levels below $n_2'$. Then, knowing $\vec{c}_{12}$, $U_1$ calculates $\vec{c}_{23}$ from $\vec{c}_{12}\oplus\vec{c}_{23}$ and uses $\vec{c}_{23}$ to obtain $\vec{c}_{31}$ from $\vec{c}_{23}\oplus\vec{c}_{31}$. $U_2$ calculates $\vec{c}_{12}$ from $\vec{c}_{12}\oplus\vec{c}_{23}$ using its knowledge of $\vec{c}_{23}$. $U_3$ receives $\vec{c}_{23}\oplus\vec{c}_{31}$ as long as $R_{123}^\sc\leq n_3'$. Assuming $R_{123}^\sc\leq n_3'$, user 3 extracts $\vec{c}_{23}$ using its knowledge of $\vec{c}_{31}$. 

Thus, as long as $R_{123}^\sc\leq n_3'$ and $2R_{123}^\sc\leq n_2'$, then $2R_{123}^\sc$ levels are sufficient for communicating all $3R_{123}^\sc$ bits of cyclic communication, for an average of 3/2 bits per level. But these inequalities hold as long as $\vec{R}\in\overline{\mathcal{C}}_d$ since
\begin{align}
R_{123}^\sc&\stackrel{\footnotesize{\eqref{D}}}{\leq} R_{31}'\\
&\stackrel{\footnotesize{\eqref{Rp}}}{=}R_{31}-R_{13}^\sb\\
&\stackrel{\footnotesize{\eqref{CS3}}}{\leq} n_3-R_{32}-R_{13}^\sb\\
\label{R123cB1}
&\stackrel{\footnotesize{\eqref{R1}}}{=}n_3-R_{23}^\sb-R_{13}^\sb
\end{align}
Moreover,
\begin{align}
2R_{123}^\sc&\stackrel{\footnotesize{\eqref{D}}}{\leq} R_{31}'+R_{23}'\\
&\stackrel{\footnotesize{\eqref{Rp}}}{=}R_{31}+R_{23}-R_{13}^\sb-R_{23}^\sb\\
&\stackrel{\footnotesize{\eqref{TRB4}}}{\leq} n_2-R_{21}-R_{13}^\sb-R_{23}^\sb\\
&\stackrel{\footnotesize{\eqref{R1}}}{=}n_2-R_{12}^\sb-R_{13}^\sb-R_{23}^\sb\\
\label{R123cB2}
&\stackrel{\footnotesize{\eqref{n2p}}}{=}n_2'\\
\label{R123cB3}
&\leq2n_2'
\end{align}
Thus, from \eqref{R123cB1} and \eqref{R123cB3} we have $R_{123}^\sc\leq\min\{n_3-R_{13}^\sb-R_{23}^\sb,n_2'\}=n_3'$ and from \eqref{R123cB2} we have $2R_{123}^\sc\leq n_2'$, and therefore, there are enough levels in the DYC$(n_1',n_2',n_3')$ for serving all bits of cyclic information flow. 

\begin{remark}
Here, we have chosen to repeat $U_2$'s transmission. Similarly, one can repeat $U_1$'s or $U_3$'s transmission instead. However, notice that repeating $U_2$'s transmission allows us to allocate one chunk of bits to levels in $\{1,\dots,n_3'\}$, and the other to levels in $\{n_3'+1,\dots,n_2'\}$ (cf. Figure \ref{RelayLevels2}), while the other options force us to allocate both chunks to $\{1,\dots,n_3'\}$, either in the uplink or in the downlink. Thus, repeating $U_2$'s transmission allows a more efficient use of the levels at the relay.
\end{remark}

For the second possibility, i.e. $1\to3\to2\to1$, a similar strategy can be used. All users send with the same rate $R_{132}^\sc$. $U_1$ repeats a bit on two levels such that, one of the levels is also used by $U_2$ at the relay, and the other is used by $U_3$. That is, one level must be in $\{1,\dots,n_2'\}$ and the other in $\{1,\dots,n_3'\}$ at the relay. Using \eqref{TRB3}, \eqref{TRB5}, and  \eqref{CS3} we can show that the levels at the relay $(n_1',n_2',n_3')$ are sufficient for this communication (keeping \eqref{R2} in mind). The assignment of the levels at the relay in this case is shown in Figures \ref{RelayLevels2} and \ref{RelayLevels3}.

\begin{figure*}[t]
\centering
\subfigure[Cyclic communication: $1\to2\to3\to1$.]{
\includegraphics[width=0.35\columnwidth]{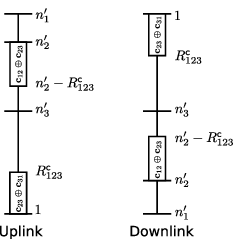}
\label{RelayLevels2}
}
\hspace{2cm}
\subfigure[Cyclic communication: $1\to3\to2\to1$.]{
\includegraphics[width=0.35\columnwidth]{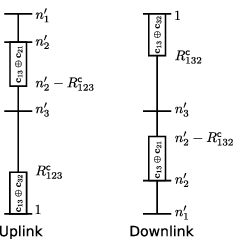}
\label{RelayLevels3}
}
\caption{Assignment of levels at the relay.}
\end{figure*}

After this stage, the rate tuple that still needs to be achieved is
\begin{align}
\label{Rpp}
\vec{R}''&\triangleq(R_{12}'',R_{13}'',R_{21}'',R_{23}'',R_{31}'',R_{32}'')
\end{align}
where
\begin{align}
R_{12}''&=R_{12}'-R_{123}^\sc,\quad R_{13}''=R_{13}'-R_{132}^\sc,\\
R_{21}''&=R_{21}'-R_{132}^\sc,\quad R_{23}''=R_{23}'-R_{123}^\sc,\\
R_{31}''&=R_{31}'-R_{123}^\sc,\quad R_{32}''=R_{32}'-R_{132}^\sc,
\end{align}
which must be achieved over a DYC$(n_1'',n_2'',n_3'')$ where 
\begin{align}
\label{n1pp}
n_1''&=n_1'-2R_{123}^\sc-2R_{132}^\sc,\\
\label{n2pp}
n_2''&=n_2'-2R_{123}^\sc-2R_{132}^\sc.
\end{align}
If $R_{123}^\sc+R_{132}^\sc\leq n_2'-n_3'$ then cyclic communication consumes only $R_{123}^\sc+R_{132}^\sc$ levels in $\{1,\dots,n_3'\}$\footnote{Recall that either $R_{123}^\sc$ or $R_{132}^\sc$, or both equal to zero. Hence $R_{123}^\sc+R_{132}^\sc$ either equals $R_{123}^\sc$ or $R_{132}^\sc$ or 0. The sum $R_{123}^\sc+R_{132}^\sc$ takes all three cases into account, i.e., the cycle $1\to2\to3\to1$, the cycle $1\to3\to2\to1$, and the case of no cyclic information flow at all.}. Thus the number of remaining levels $n_3''$ is $$n_3''=n_3'-R_{123}^\sc-R_{132}^\sc.$$ Otherwise, if $n_2'-n_3'$ is less than $R_{123}^\sc+R_{132}^\sc$, then the upper chunk of bits in the uplink in Figures \ref{RelayLevels2} and \ref{RelayLevels3} occupies levels in $\{1,\dots,n_3'\}$. The remaining $n_3''$ levels shared by all users is then given by
\begin{align}
n_3''&=n_3'-R_{123}^\sc-R_{132}^\sc-(R_{123}^\sc+R_{132}^\sc-n_2'+n_3')\\
&=n_2''.
\end{align}
Thus, we can write
\begin{align}
\label{n3pp}
n_3''=\min\{n_3'-R_{123}^\sc-R_{132}^\sc,n_2''\}.
\end{align}
Recall that either $R_{123}^\sc=R_{132}^\sc=0$, or $R_{123}^\sc=0$ and $R_{132}^\sc>0$, or $R_{123}^\sc>0$ and $R_{132}^\sc=0$. After considering cases $\sb$) and $\sc$), only case $\su$) with uni-directional information flow remains.

\subsection{Uni-directional information flow over the DYC}
\label{AUC}
Finally, we are left with a rate tuple $\vec{R}''$ to achieve over a DYC$(n_1'',n_2'',n_3'')$ where at least 3 components of $\vec{R}''$ are zero. The non-zero components of $\vec{R}''$ fall neither into cases $\sb$) nor $\sc$), i.e., neither bi-directional nor cyclic communication. We have 6 different possibilities, depending on the positions of the zero components of $\vec{R}''$. We describe one of these possibilities in details, the rest are similar and are described briefly in Appendix \ref{App:UniDirectionalCases} on page \pageref{App:UniDirectionalCases}.

Consider the scenario where $R_{21}''=R_{31}''=R_{32}''=0$. In this case, the non-zero components of $\vec{R}''$ are $R_{12}''$, $R_{13}''$, and $R_{23}''$, or a subset thereof. Let $$\vec{u}_{12}\in\mathbb{F}_2^{R_{12}''},\quad\vec{u}_{13}\in\mathbb{F}_2^{R_{13}''},\quad\text{and}\quad\vec{u}_{23}\in\mathbb{F}_2^{R_{23}''}$$ denote the binary vectors to be communicated. In the uplink, $U_1$ uses levels $\{n_1''-R_{12}''+1,\dots,n_1''\}$ to send $\vec{u}_{12}$ and levels $\{n_1''-R_{12}''-R_{13}''+1,\dots,n_1''-R_{12}''\}$ to send $\vec{u}_{13}$, and $U_2$ uses levels $\{1,\dots,R_{23}''\}$ to send $\vec{u}_{23}$ to the relay. The relay then forwards $\vec{u}_{13}$ on levels $\{1,\dots,R_{13}''\}$, $\vec{u}_{23}$ on levels $\{R_{13}''+1,\dots,R_{13}''+R_{23}''\}$, and $\vec{u}_{12}$ on levels $\{n_2''-R_{12}''+1,\dots,n_2''\}$ as shown in Figure \ref{RelayLevels4}. Note that while each two bits of bi-directional information flow consume 1 level, and each 3 bits of cyclic information flow consume 2 levels, in the uni-directional case, each bit consumes one level.

The uni-directional strategy works for communicating all $R_{12}''+R_{13}''+R_{23}''$ bits of uni-directional information flow if the following inequalities are satisfied in the uplink

\begin{figure}[t]
\centering
\includegraphics[height=0.35\columnwidth]{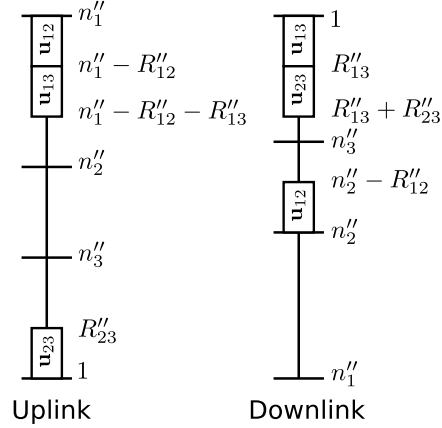}
\caption{Acyclic uni-directional communication over the DYC.}
\label{RelayLevels4}
\end{figure}

\begin{align}
\label{R23''}
R_{23}''&\leq n_2''\\
R_{12}''+R_{13}''+R_{23}''&\leq n_1''
\end{align}
and the following in the downlink
\begin{align}
R_{23}''&\leq n_3''\\
R_{23}''+R_{13}''&\leq n_3''\\
\label{R12''pR13''pR23''}
R_{12}''+R_{13}''+R_{23}''&\leq n_2''.
\end{align}
Combining \eqref{R23''}-\eqref{R12''pR13''pR23''}, we get 
\begin{align}
\label{InEq1}
R_{23}''+R_{13}''&\leq n_3''\\
\label{InEq2}
R_{12}''+R_{13}''+R_{23}''&\leq n_2''.
\end{align}

But these inequalities are satisfied as long as $\vec{R}\in\overline{\mathcal{C}}_d$. To see this, consider the first inequality \eqref{InEq1},
\begin{align}
\label{R23ppPR13pp1}
R_{23}''+R_{13}''&\stackrel{\footnotesize{\eqref{Rpp}}}{=}R_{23}'+R_{13}'-R_{123}^\sc-R_{132}^\sc\\
\label{R23ppPR13pp2}
&\stackrel{\footnotesize{\eqref{Rp}}}{=}R_{23}+R_{13}-R_{23}^\sb-R_{13}^\sb-R_{123}^\sc-R_{132}^\sc\nonumber\\
&\stackrel{\footnotesize{\eqref{CS3}}}{\leq} n_3-R_{23}^\sb-R_{13}^\sb-R_{123}^\sc-R_{132}^\sc.
\end{align}
Recall from \eqref{R1} and \eqref{R2} that either $R_{123}^\sc$ or $R_{132}^\sc$ must be zero. If $R_{132}^\sc=0$ then starting from \eqref{R23ppPR13pp2} we get
\begin{align}
R_{23}''+R_{13}''&\stackrel{\footnotesize{\eqref{R23ppPR13pp2}}}{=}R_{23}+R_{13}-R_{23}^\sb-R_{13}^\sb-R_{123}^\sc\\
&\stackrel{\footnotesize{\eqref{TRB5}}}{\leq} n_2-R_{12}-R_{23}^\sb-R_{13}^\sb-R_{123}^\sc\\
&\stackrel{}{=} n_2-R_{12}+R_{12}^\sb-R_{12}^\sb-R_{23}^\sb-R_{13}^\sb-R_{123}^\sc\nonumber\\
&\stackrel{\footnotesize{\eqref{D}}}{=} n_2-R_{12}^\sb-R_{23}^\sb-R_{13}^\sb-2R_{123}^\sc\\
&\stackrel{\footnotesize{\eqref{n2p}}}{\leq} n_2'-2R_{123}^\sc.
\end{align}
If $R_{132}^\sc>0$, then $R_{123}^\sc=0$, $R_{12}^\sb=R_{12}$, and $R_{23}^\sb=R_{23}$ by \eqref{R2} and thus
\begin{align}
R_{23}''+R_{13}''&\stackrel{\footnotesize{\eqref{R23ppPR13pp2}}}{=}R_{23}+R_{13}-R_{13}^\sb-R_{23}^\sb-R_{132}^\sc\\
&\stackrel{\footnotesize{\eqref{R2}}}{=} R_{13}-R_{13}^\sb-R_{132}^\sc\\
&\stackrel{\footnotesize{\eqref{TRB1}}}{\leq} n_2-R_{12}-R_{32}-R_{13}^\sb-R_{132}^\sc\\
&\stackrel{}{=} n_2-R_{12}^\sb-R_{32}-R_{13}^\sb-R_{132}^\sc\\
&\stackrel{}{=} n_2-R_{12}^\sb-R_{32}+R_{23}^\sb-R_{23}^\sb-R_{13}^\sb-R_{132}^\sc\nonumber\\
&\stackrel{\footnotesize{\eqref{E}}}{=} n_2-R_{12}^\sb-R_{23}^\sb-R_{13}^\sb-2R_{132}^\sc\\
&\stackrel{\footnotesize{\eqref{n2p}}}{\leq} n_2'-2R_{132}^\sc.
\end{align}
Thus, since either $R_{123}^\sc=0$ or $R_{132}^\sc=0$ we get
\begin{align}
R_{23}''+R_{13}''&\leq \min\{n_3-R_{23}^\sb-R_{13}^\sb-R_{123}^\sc-R_{132}^\sc,\nonumber\\
&\hspace{1.5cm} n_2'-2R_{123}^\sc-2R_{132}^\sc\}\\
&\stackrel{\footnotesize{\eqref{n3p}}}{=} \min\{n_3'-R_{123}^\sc-R_{132}^\sc, n_2'-2R_{123}^\sc-2R_{132}^\sc\}\nonumber\\
&\stackrel{\footnotesize{\eqref{n3pp}}}{=}n_3''.
\end{align}

Thus, \eqref{InEq1} is satisfied. Consider now the second inequality \eqref{InEq2}. From \eqref{Rp} and \eqref{Rpp} we have
\begin{align*}
R_{12}''+R_{13}''+R_{23}''&\leq R_{12}+R_{13}+R_{23}-R_{12}^\sb-R_{13}^\sb-R_{23}^\sb\nonumber\\
&\quad-2R_{123}^\sc-R_{132}^\sc.
\end{align*}
If $R_{132}^\sc=0$ then
\begin{align}
R_{12}''+R_{13}''+R_{23}''&\leq R_{12}+R_{13}+R_{23}-R_{12}^\sb-R_{13}^\sb\nonumber\\
&\quad-R_{23}^\sb-2R_{123}^\sc\nonumber\\
\label{3R12}
&\stackrel{\footnotesize{\eqref{TRB5}}}{\leq} n_2-R_{12}^\sb-R_{13}^\sb-R_{23}^\sb-2R_{123}^\sc.
\end{align}
Otherwise, if $R_{132}^\sc>0$ then $R_{123}^\sc=0$. In this case, by using $R_{23}^\sb=R_{23}$ from \eqref{R2} we get
\begin{align}
R_{12}''+R_{13}''+R_{23}''&\leq R_{12}+R_{13}-R_{12}^\sb-R_{13}^\sb-R_{132}^\sc\\
&\stackrel{\footnotesize{\eqref{TRB1}}}{\leq} n_2-R_{32}-R_{12}^\sb-R_{13}^\sb-R_{132}^\sc\\
&\stackrel{}{=} n_2-R_{32}+R_{23}^\sb-R_{23}^\sb-R_{12}^\sb\nonumber\\
&\quad-R_{13}^\sb-R_{132}^\sc\\
\label{3R23}
&\stackrel{\footnotesize{\eqref{E}}}{\leq} n_2-R_{12}^\sb-R_{13}^\sb-R_{23}^\sb-2R_{132}^\sc.
\end{align}
Combining \eqref{3R12} and \eqref{3R23}, keeping in mind that either $R_{123}^\sc=0$ or $R_{132}^\sc=0$, we obtain
\begin{align}
R_{12}''+R_{13}''+R_{23}''&\leq n_2-R_{12}^\sb-R_{13}^\sb-R_{23}^\sb\nonumber\\
&\quad-2R_{123}^\sc-2R_{132}^\sc\\
&\stackrel{\footnotesize{\eqref{n2p}}}{=}n_2'-2R_{123}^\sc-2R_{132}^\sc\\
&\stackrel{\footnotesize{\eqref{n2pp}}}{=}n_2''.
\end{align}
As a result, inequality \eqref{InEq2} is satisfied, and there exist enough levels for communicating all $R_{12}''+R_{13}''+R_{23}''$ bits.

Five other possibilities of uni-directional information flow remain. These cases are similar to the case studied above, and are considered briefly in Appendix \ref{App:UniDirectionalCases}. Consequently, after assigning levels for bi-directional communication and for cyclic communication, enough levels remain to communicate all the remaining bits in $\vec{R}''$. We obtain the following theorem.

\begin{theorem}
\label{IB}
Every rate tuple $\vec{R}\in\mathbb{N}^6\cap\overline{\mathcal{C}}_d$ is achievable.
\end{theorem}

\begin{proof}
We assign one level for each 2 bits of bi-directional communication if any, as described in Section \ref{BDC}. Then, we assign two levels for each three bits of cyclic communication if any, as described in Section \ref{CC}. The remaining bits of uni-directional communication are communicated as described in Section \ref{AUC}. Since the levels at the relay are enough for this strategy as long as $\vec{R}$ is integer and belongs to $\overline{\mathcal{C}}_d$ as shown above, the result follows.
\end{proof}

Now, we use this result to show that any rate tuple in $\overline{\mathcal{C}}_d$, not necessarily integer, is achievable. It was shown in \cite{AvestimehrKhajehnejadSezginHassibi} that a $Q$-symbol extension of a multi-pair BRC ($Q$ time slots) is can be modeled as a multi-pair BRC where the channel parameters (number of levels) is the same as that of the original network multiplied by $Q$. The same statement holds here. We can think of a DYC$(n_1,n_2,n_3)$ over $Q$ time slots as a DYC$(Qn_1,Qn_2,Qn_3)$. Now we can state the following theorem.

\begin{theorem}
The capacity region $\mathcal{C}_d$ of the DYC is $\overline{\mathcal{C}}_d$.
\end{theorem}
\begin{proof}
Since all inequalities representing the boundary of the outer bound $\overline{\mathcal{C}}_d$, i.e. \eqref{TRB1}-\eqref{TRB6} and \eqref{CS3}, have integer coefficients, then all the corner points of the outer bound are rational. Consider a corner point $$\vec{R}=\left(\frac{P_{12}}{Q_{12}},\frac{P_{13}}{Q_{13}},\frac{P_{21}}{Q_{21}},\frac{P_{23}}{Q_{23}},\frac{P_{31}}{Q_{31}},\frac{P_{32}}{Q_{32}}\right),$$
where $P_{jk},Q_{jk}\in\mathbb{N}$ for all $j,k\in\{1,2,3\}$, $j\neq k$. This corner point is achievable as follows. Use $Q$ time slots to achieve the rate tuple $Q\vec{R}$ where $$Q=\prod_{j=1}^3\prod_{\substack{k=1\\k\neq j}}^3Q_{jk},$$
over a DYC$(Qn_1,Qn_2,Qn_3)$. Since $\vec{R}\in\overline{\mathcal{C}}_d$ then $Q\vec{R}\in\overline{\mathcal{C}}'_d$ where $\overline{\mathcal{C}}'_d$ is the outer bound of Theorem \ref{OB} for a DYC$(Qn_1,Qn_2,Qn_3)$. Moreover, $Q\vec{R}\in\mathbb{N}^6$. Thus $Q\vec{R}\in\mathbb{N}^6\cap\overline{\mathcal{C}}'_d$ which means that it is achievable according to Theorem \ref{IB}. But $Q\vec{R}$ being achievable in a DYC$(Qn_1,Qn_2,Qn_3)$ implies that $\vec{R}$ is achievable in a DYC$(n_1,n_2,n_3)$. Therefore all corner points of $\overline{\mathcal{C}}_d$ are achievable, even if not integer valued. All other points in $\overline{\mathcal{C}}_d$ are achievable by time sharing between different corner points and the statement of the theorem follows.
\end{proof}

Having established the capacity region of the DYC, we can now extend these results to the Gaussian case, GYC, by using the insights we gained from the DYC. In the following section, we provide an outer bound on $\mathcal{C}_g$, which we show later to be achievable within a constant gap which is independent of channel parameters.

\section{The GYC: An Outer Bound}
\label{Sec:GYCOuterBound}
Upper bounds on achievable rates in the GYC were given in \cite{ChaabanSezginAvestimehr_YC_SC} using the cut-set bounds and genie aided bounds. In this section, we briefly summarize those bounds. We will make use of these bounds to obtain a constant gap characterization of the capacity of the GYC. First, consider the cut-set bounds in \cite[Corollary 1]{ChaabanSezginAvestimehr_YC_SC} which provide upper bounds on the sum of two components of $\vec{R}$ given by
\begin{align}
\label{2RB1}
R_{31}+R_{32}&\leq C(h_3^2P)\\
R_{13}+R_{23}&\leq C(h_3^2P)\\
R_{21}+R_{23}&\leq C(h_2^2P)\\
R_{12}+R_{32}&\leq C(h_2^2P)\\
R_{12}+R_{13}&\leq \min\{C(h_1^2P),C(h_2^2P+h_3^2P)\}\\
\label{2RB6}
R_{21}+R_{31}&\leq \min\{C(h_1^2P),C((|h_2|+|h_3|)^2P)\},
\end{align}
where $C(x)=\frac{1}{2}\log(1+x)$ as defined earlier. In the same paper \cite{ChaabanSezginAvestimehr_YC_SC}, it was shown that these bounds are very loose in terms of sum-capacity as $P$ increases. Additional bounds on the sum of three components of $\vec{R}$, similar to those in Lemmas \ref{FromRelay} and \ref{GUB}, are required for a constant gap characterization of the sum-capacity. These bounds, given in \cite[Lemmas 1 and 2]{ChaabanSezginAvestimehr_YC_SC} can be written as follows

\begin{align}
\label{3RB1}
R_{12}+R_{13}+R_{32}&\leq C(h_2^2P+h_3^2P)\\
R_{12}+R_{13}+R_{23}&\leq C(h_2^2P+h_3^2P)\\
R_{21}+R_{23}+R_{13}&\leq C(h_1^2P+h_3^2P)\\
R_{21}+R_{23}+R_{31}&\leq C((|h_2|+|h_3|)^2P)\\
R_{31}+R_{32}+R_{12}&\leq C(h_1^2P+h_2^2P)\\
\label{3RB6}
R_{31}+R_{32}+R_{21}&\leq C((|h_2|+|h_3|)^2P).
\end{align}

These bounds, combined all together, provide an outer bound on the capacity region of the GYC. Let us denote this outer bound by $\overline{\mathcal{C}}_g$ which is given by
\begin{align}
\overline{\mathcal{C}}_g=\left\{\vec{R}\in\mathbb{R}_+^6|\  \eqref{2RB1}-\eqref{2RB6} \text{ and } \eqref{3RB1}-\eqref{3RB6} \text{ are satisfied}\right\}.
\end{align} 

\begin{theorem}
The capacity region $\mathcal{C}_g$ of the GYC is outer bounded by $\overline{\mathcal{C}}_g$, $$\mathcal{C}_g\subseteq\overline{\mathcal{C}}_g.$$
\end{theorem}

\section{The GYC: An Achievable Scheme}
\label{Sec:GYCInnerBound}
In this section, we provide an achievable scheme for the GYC, and consequently an inner bound on $\mathcal{C}_g$. This inner bound is achieved by using a scheme similar to the one for the DYC from the previous sections adapted to the Gaussian case. Namely, this scheme utilizes network coding realized with lattice codes \cite{Loeliger}. We start with a brief introduction about lattice codes (more details can be found in \cite{NazerGastpar}), before proceeding to describe the achievable scheme.

\subsection{Lattice Codes}
An $n$-dimensional lattice $\Lambda$ is a subset of $\mathbb{R}^n$ such that $\lambda_1,\lambda_2\in\Lambda\Rightarrow\lambda_1+\lambda_2\in\Lambda$, i.e. it is an additive subgroup of $\mathbb{R}^n$. The fundamental Voronoi region $\mathcal{V}(\Lambda)$ of $\Lambda$ is the set of all points in $\mathbb{R}^n$ whose distance to the origin is smaller that that to any other $\lambda\in\Lambda$. Thus, by quantizing points in $\mathbb{R}^n$ to their closest lattice point, all points in $\mathcal{V}(\Lambda)$ are mapped to the all zero vector.

In this work, we need nested lattice codes. Two lattices are required for nested lattice codes, a coarse lattice $\Lambda_c$ and a fine lattice $\Lambda_f$ where $\Lambda_c\subseteq\Lambda_f$. We denote a nested lattice code by the pair $(\Lambda_f,\Lambda_c)$. The codewords are chosen as the fine lattice points $\lambda_f\in\Lambda_f$ that lie in $\mathcal{V}(\Lambda_c)$. The power constraint is satisfied by an appropriate choice of $\Lambda_c$ and the rate of the code is defined by the number of fine lattice points in $\Lambda_f\cap\mathcal{V}(\Lambda_c)$. In the sequel, we are going to need the following result from \cite{NarayananPravinSprintson}. 

Given two nodes A and B, with messages $u_A$ and $u_B$, respectively, where both messages have rate $R$. The two nodes use the same nested lattice codebook $(\Lambda_f,\Lambda_c)$ with power $P$, to encode their messages into codewords $\lambda_A$ and $\lambda_B$ with length $n$. The nodes then construct their transmit signals $x_A^n$ and $x_B^n$ as
\begin{align}
x_A^n&=(\lambda_A-d_A)\bmod\Lambda_c\\
x_B^n&=(\lambda_B-d_B)\bmod\Lambda_c
\end{align}
where $d_A$ and $d_B$ are $n$-dimensional dither vectors uniformly distributed over $\mathcal{V}(\Lambda_c)$, known at the relay and nodes A and B\footnote{\alert{The $x\bmod\Lambda$ operation returns the quantization error corresponding to quantizing $x$ to the nearest point in the lattice $\Lambda$ \cite{NazerGastpar}}}. A relay nodes receives $$y_R^n=x_A^n+x_B^n+z_R^n$$ where $z_R^n$ is an additive white Gaussian noise with i.i.d. components with zero mean and variance $\sigma^2$. 
\begin{lemma}[\cite{NarayananPravinSprintson}]
\label{Lemma:AlignDecodeRate}
The relay can decode the sum $$(\lambda_A+\lambda_B)\bmod \Lambda_c$$ from $y_R^n$ reliably as long as $$R\leq\frac{1}{2}\log\left(\frac{1}{2}+\frac{P}{\sigma^2}\right).$$
\end{lemma}
\begin{lemma}[\cite{NarayananPravinSprintson}]
\label{Lemma:AlignDecode}
Node A, knowing $(\lambda_A+\lambda_B)\bmod \Lambda_c$ and $\lambda_A$, can extract $\lambda_B$ and hence also $u_B$.
\end{lemma}

Let us now describe the achievable scheme in the uplink.

\subsection{Uplink}
In the uplink, $U_i$ splits each message $m_{ij}\in\{1,\dots,2^{nR_{ij}}\}$ into a bi-directional communication part $m_{ij}^\sb\in\{1,\dots,2^{R_{ij}^\sb}\}$, a cyclic communication part $m_{ij}^\sc\in\{1,\dots,2^{R_{ij}^\sc}\}$, and a uni-directional communication part $m_{ij}^\su\in\{1,\dots,2^{R_{ij}^\su}\}$. Indeed, $m_{ij}^\sb$ is a uni-directional message in the sense that it is sent from $U_i$ and intended to $U_j$. The superscript $\sb$ here is used to indicate that a bi-directional communication strategy will be used to communicate this message. Same goes for the superscripts $\sc$ and $\su$. The rates of the bi-directional communication messages at different users are chosen such that
\begin{align}
\label{BiDRates}
R_{12}^\sb=R_{21}^\sb,\quad R_{13}^\sb=R_{31}^\sb,\quad \text{and }R_{23}^\sb=R_{32}^\sb.
\end{align}
Moreover, the rates of the cyclic streams are chosen to satisfy
\begin{align}
\label{CycRates}
R_{12}^\sc=R_{23}^\sc=R_{31}^\sc\triangleq R_{123}^\sc,\quad \text{and }R_{13}^\sc=R_{32}^\sc=R_{21}^\sc\triangleq R_{132}^\sc.
\end{align}
Thus, the rate of a message $m_{ij}$ is split into three parts, $R_{ij}^\sb$, $R_{ij}^\sc$, $R_{ij}^\su$, such that $R_{ij}=R_{ij}^\sb+R_{ij}^\sc+R_{ij}^\su$. In the following, we describe the encoding of the bi-directional, cyclic, and uni-directional streams to establish an inner bound on $\mathcal{C}_g$.

\subsubsection{Encoding of bi-directional streams}
\label{Sec:BiDirStreams}
The users use n-dimensional nested lattices to encode the bi-directional communication streams. Let us consider the bi-directional communication between users 1 and 2, i.e., the messages $m_{12}^\sb$ and $m_{21}^\sb$. \begin{itemize}
\item $U_2$ uses a nested lattice code $(\Lambda_{21}^\sb,\Lambda_{21,c}^\sb)$. The rate of the code is $R_{21}^\sb$ and the power is $\alpha_{21}^\sb P$. Each message $m_{21}^\sb$ is mapped into a lattice point $\lambda_{21}^\sb\in\Lambda_{21}^\sb\cap\mathcal{V}(\Lambda_{21,c}^\sb)$. 
\item $U_1$ uses a scaled version of the lattice code used for $m_{21}^\sb$ to encode $m_{12}^\sb$. The code for $m_{12}^\sb$ is designed such that the bi-directional communication signals align at the relay. That is, $U_1$ uses a lattice code $(\Lambda_{12}^\sb,\Lambda_{12,c}^\sb)=\left(\frac{h_2}{h_1}\Lambda_{21}^\sb,\frac{h_2}{h_1}\Lambda_{21,c}^\sb\right)$, for encoding $m_{12}^\sb$. Each $m_{12}^\sb$ is mapped into a point in $\Lambda_{12}^\sb\cap\mathcal{V}(\Lambda_{12,c}^\sb)$, denoted $\lambda_{12}^\sb$. 
\end{itemize}

Notice that the rate of the lattice code $\Lambda_{12}^\sb$ is $R_{12}^\sb=R_{21}^\sb$, and the power of $\Lambda_{12}^\sb$ is 
\begin{align}
\label{A12bVsA21b}
\alpha_{12}^\sb P=\frac{h_2^2}{h_1^2}\alpha_{21}^\sb P.
\end{align}
In this way, the lattice codewords $\lambda_{12}^\sb$ and $\lambda_{21}^\sb$ align at the relay, since $$h_1\lambda_{12}^\sb+h_2\lambda_{21}^\sb\in h_2(\Lambda_{21}^\sb\cap\mathcal{V}(\Lambda_{21,c}^\sb)).$$ Then, $U_1$ and $U_2$ add the appropriate random dither vectors to $\lambda_{12}^\sb$ and $\lambda_{21}^\sb$, to construct the transmit signals $b_{12}^n$ and $b_{21}^n$. Namely,
\begin{align}
b_{12}^n&=(\lambda_{12}^\sb-d_{12}^\sb)\bmod \Lambda_{12,c}^\sb\\
b_{21}^n&=(\lambda_{21}^\sb-d_{21}^\sb)\bmod \Lambda_{21,c}^\sb
\end{align}
where $d_{12}^\sb$ is uniformly distributed over $\mathcal{V}(\Lambda_{12,c}^\sb)$, and $d_{21}^\sb$ is uniformly distributed over $\mathcal{V}(\Lambda_{21,c}^\sb)$. Both dither vectors are known at $U_1$, $U_2$, and the relay (see \cite{NarayananPravinSprintson, GunduzYenerGoldsmithPoor_IT}).

A similar procedure is done for encoding $m_{31}^\sb$, $m_{13}^\sb$, $m_{32}^\sb$, and $m_{23}^\sb$ into 
\begin{align}
b_{31}^n&=(\lambda_{31}^\sb-d_{31}^\sb)\bmod \Lambda_{31,c}^\sb\\
b_{13}^n&=(\lambda_{13}^\sb-d_{13}^\sb)\bmod \Lambda_{13,c}^\sb\\
b_{32}^n&=(\lambda_{32}^\sb-d_{32}^\sb)\bmod \Lambda_{32,c}^\sb\\
b_{23}^n&=(\lambda_{23}^\sb-d_{23}^\sb)\bmod \Lambda_{23,c}^\sb
\end{align}
with powers $\alpha_{31}^\sb P$, $\alpha_{13}^\sb P$, $\alpha_{32}^\sb P$ and $\alpha_{23}^\sb P$, where the lattice codebooks used are $(\Lambda_{31}^\sb,\Lambda_{31,c}^\sb)$, $(\Lambda_{13}^\sb,\Lambda_{13,c}^\sb)$, $(\Lambda_{32}^\sb,\Lambda_{32,c}^\sb)$, and $(\Lambda_{23}^\sb,\Lambda_{23,c}^\sb)$,
respectively, where
\begin{align}
(\Lambda_{13}^\sb,\Lambda_{13,c}^\sb)&=\left(\frac{h_3}{h_1}\Lambda_{31}^\sb,\frac{h_3}{h_1}\Lambda_{31,c}^\sb\right),\\
(\Lambda_{23}^\sb,\Lambda_{23,c}^\sb)&=\left(\frac{h_3}{h_2}\Lambda_{32}^\sb,\frac{h_3}{h_2}\Lambda_{32,c}^\sb\right),
\end{align}
so that
\begin{align}
\label{A13bVsA31b}
\alpha_{13}^\sb P=\frac{h_3^2}{h_1^2}\alpha_{31}^\sb P,\\
\label{A23bVsA32b}
\alpha_{23}^\sb P=\frac{h_3^2}{h_2^2}\alpha_{32}^\sb P.
\end{align}

\subsubsection{Encoding of cyclic streams}
Consider the messages $m_{12}^\sc$, $m_{23}^\sc$ and $m_{31}^\sc$ constituting the cycle $1\to2\to3\to1$. To communicate these messages, the second user sends $m_{23}^\sc$ encoded in two different signals, one of them aligned with the signal sent by the first user (corresponding to $m_{12}^\sc$), and one aligned with that sent by the third user (corresponding to $m_{31}^\sc$). Notice that this mimics the scheme used to achieve cyclic information flow over the DYC in section \ref{CC}. Here, the alignment is also guaranteed using nested lattices in a similar way as for the bi-directional streams in Section \ref{Sec:BiDirStreams}.

To this end, let $U_2$ use a nested lattice code $(\Lambda_{23}^\sc,\Lambda_{23,c}^\sc)$ with rate $R_{123}^\sc$, and the appropriate dither vector, for encoding $m_{23}^\sc$ into a codeword
\begin{align}
c_{23}^n&=(\lambda_{23}^\sc-d_{23}^\sc)\bmod \Lambda_{23,c}^\sc
\end{align}
with power $\alpha_{23}^\sc P$. Then, $U_1$ uses a nested lattice code $(\Lambda_{12}^\sc,\Lambda_{12,c}^\sc)$ where $\Lambda_{12}^\sc=\frac{h_2}{h_1}\Lambda_{23}^\sc$ and $\Lambda_{12,c}^\sc=\frac{h_2}{h_1}\Lambda_{23,c}^\sc$, and the appropriate dither vector, to encode $m_{12}^\sc$ to 
\begin{align}
c_{12}^n&=(\lambda_{12}^\sc-d_{12}^\sc)\bmod \Lambda_{12,c}^\sc
\end{align}
with power 
\begin{align}
\label{A12cVsA23c}
\alpha_{12}^\sc P=\frac{h_2^2}{h_1^2}\alpha_{23}^\sc P.
\end{align}
Notice that this ensures alignment of the codes $(\Lambda_{23}^\sc,\Lambda_{23,c}^\sc)$ and $(\Lambda_{12}^\sc,\Lambda_{12,c}^\sc)$ at the relay as in Section \ref{Sec:BiDirStreams}. 

$U_3$ uses a nested lattice code $(\Lambda_{31}^\sc,\Lambda_{31,c}^\sc)$ with rate $R_{123}^\sc$, and the appropriate dither vector, for encoding $m_{31}^\sc$ into a codeword 
\begin{align}
c_{31}^n&=(\lambda_{31}^\sc-d_{31}^\sc)\bmod \Lambda_{31,c}^\sc
\end{align}
with power $\alpha_{31}^\sc P$. $U_2$ encodes $m_{23}^\sc$ again into 
\begin{align}
\tilde{c}_{23}^n&=(\tilde{\lambda}_{23}^\sc-\tilde{d}_{23}^\sc)\bmod \tilde{\Lambda}_{23,c}^\sc
\end{align}
using a nested lattice code $(\tilde{\Lambda}_{23}^\sc,\tilde{\Lambda}_{23,c}^\sc)$ where $\tilde{\Lambda}_{23}^\sc=\frac{h_3}{h_2}\Lambda_{31}^\sc$ and $\tilde{\Lambda}_{23,c}^\sc=\frac{h_3}{h_2}\Lambda_{31,c}^\sc$ with power 
\begin{align}
\label{A23cVsA31c}
\tilde{\alpha}_{23}^\sc P=\frac{h_3^2}{h_2^2}\alpha_{31}^\sc P.
\end{align} 
This ensures the alignment of $(\Lambda_{31}^\sc,\Lambda_{31,c}^\sc)$ and  $(\tilde{\Lambda}_{23}^\sc,\tilde{\Lambda}_{23,c}^\sc)$ at the relay.

Similar encoding is performed on the messages of the other cycle $1\to3\to2\to1$ where the first user encodes $m_{13}^\sc$ twice, namely into 
\begin{align}
c_{13}^n&=(\lambda_{13}^\sc-d_{13}^\sc)\bmod \Lambda_{13,c}^\sc\\
\tilde{c}_{13}^\sc&=(\tilde{\lambda}_{13}^\sc-\tilde{d}_{13}^\sc)\bmod \tilde{\Lambda}_{13,c}^\sc
\end{align}
to be aligned with 
\begin{align}
c_{32}^n&=(\lambda_{32}^\sc-d_{32}^\sc)\bmod \Lambda_{32,c}^\sc
\end{align}
(the codeword corresponding to $m_{32}^\sc$) and 
\begin{align}
c_{21}^n&=(\lambda_{21}^\sc-d_{21}^\sc)\bmod \Lambda_{21,c}^\sc
\end{align}
(the codeword corresponding to $m_{21}^\sc$), respectively. The powers of $c_{32}^n$, $c_{21}^n$, $c_{13}^n$ and $\tilde{c}_{13}^n$ are $\alpha_{32}^\sc P$, $\alpha_{21}^\sc P$, 
\begin{align}
\label{A13cVsA32c}
\alpha_{13}^\sc P=\frac{h_3^2}{h_1^2}\alpha_{32}^\sc P,
\end{align}
and 
\begin{align}
\label{A13cVsA21c}
\tilde{\alpha}_{13}^\sc P=\frac{h_2^2}{h_1^2}\alpha_{21}^\sc P,
\end{align} 
respectively, and the rates are all equal to $R_{132}^\sc$.

\subsubsection{Encoding of the uni-directional streams}
The uni-directional streams $m_{ij}^\su$ with rates $R_{ij}^\su$ are encoded using Gaussian codes. Each $m_{ij}^\su$ is mapped to a codeword $u_{ij}^n$ which is a Gaussian code whose components are i.i.d. Gaussian distributed with zero mean and variance $\alpha_{ij}^\su P$.

This encoding is depicted graphically in Figure \ref{YC_scheme} which shows both the uplink and the downlink in the GYC. Having completed the encoding of all the messages, we now proceed with the submission of the transmit signals.

\begin{figure*}
\centering
\includegraphics[width=.9\textwidth]{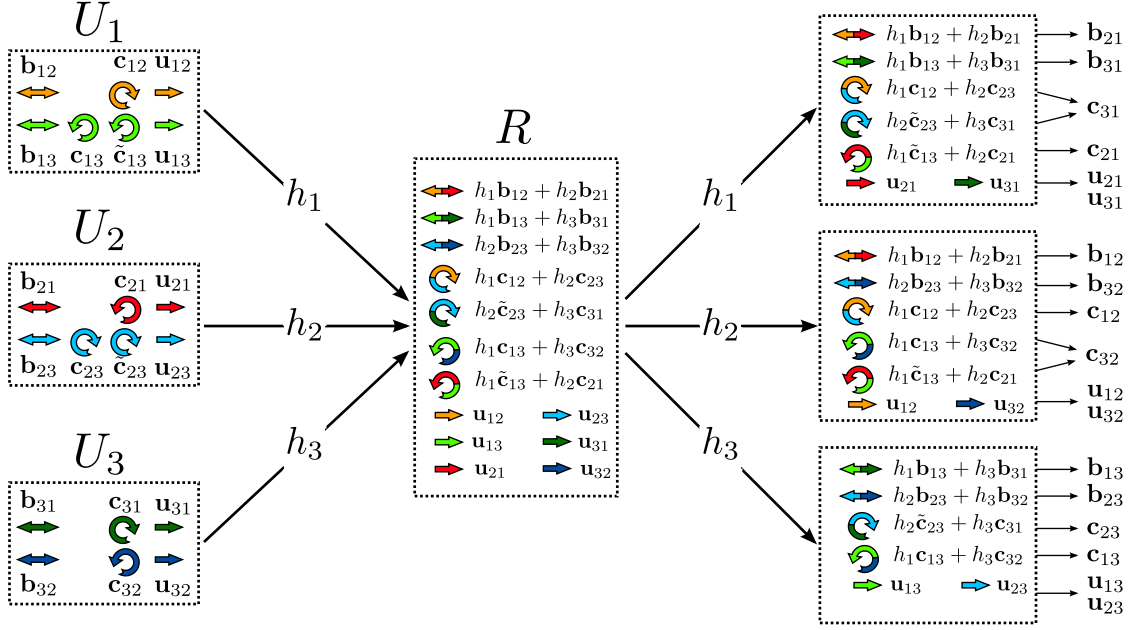}
\caption{\alert{A graphical illustration of the proposed transmit strategy for the GYC. Each user sends a bi-directional signal $\vec{b}_{jk}$, a cyclic signal $\vec{c}_{jk}$, and a uni-directional signal $\vec{u}_{jk}$ corresponding to the message $m_{jk}$. The relay computes linear combinations of the bi-directional and the cyclic signals, and decodes the uni-directional ones as shown, then it sends them to the users which decode their desired signals from the received signal using their own signals as side information.}}
\label{YC_scheme}
\hrule
\end{figure*}

\subsubsection{Transmit signals}
Each user then transmits the superposition of all its codewords as follows
\begin{align}
x_1^n&=b_{12}^n+b_{13}^n+c_{12}^n+c_{13}^n+\tilde{c}_{13}^n+u_{12}^n+u_{13}^n\\
x_2^n&=b_{21}^n+b_{23}^n+c_{21}^n+c_{23}^n+\tilde{c}_{23}^n+u_{21}^n+u_{23}^n\\
x_3^n&=b_{31}^n+b_{32}^n+c_{31}^n+c_{32}^n+u_{31}^n+u_{32}^n.
\end{align}
Recall that $b_{ij}^n$ corresponds to bi-directional streams, $c_{ij}^n$ and $\tilde{c}_{ij}^n$ correspond to cyclic streams, and $u_{ij}^n$ corresponds to uni-directional streams.
 
Notice that the transmission should not violate the power constraint at the users, i.e., 
\begin{align}
\label{Cond:Alpha1}
\alpha_{12}^\sb+\alpha_{13}^\sb+\alpha_{12}^\sc+\alpha_{13}^\sc+\tilde{\alpha}_{13}^\sc+\alpha_{12}^\su+\alpha_{13}^\su&\leq1\\
\label{Cond:Alpha2}
\alpha_{21}^\sb+\alpha_{23}^\sb+\alpha_{21}^\sc+\alpha_{23}^\sc+\tilde{\alpha}_{23}^\sc+\alpha_{21}^\su+\alpha_{23}^\su&\leq1\\
\label{Cond:Alpha3}
\alpha_{31}^\sb+\alpha_{32}^\sb+\alpha_{31}^\sc+\alpha_{32}^\sc+\alpha_{31}^\su+\alpha_{32}^\su&\leq1.
\end{align}

After transmitting $x_1^n$, $x_2^n$, and $x_3^n$ from $U_1$, $U_2$, and $U_3$, respectively, the relay decodes the transmitted signals (or their superposition) as described next.

\subsubsection{Decoding at the relay}
Decoding at the relay proceeds as follows. The uni-directional communication signals $u_{12}^n$, $u_{13}^n$, $u_{21}^n$, and $u_{23}^n$ are decoded first successively in the following order $$u_{12}^n\to u_{13}^n\to u_{21}^n\to u_{23}^n$$
while treating the remaining signals as noise. The effective noise power while decoding $u_{12}^n$ is given by
\begin{align*}
&h_3^2P(\alpha_{31}^\sb+\alpha_{32}^\sb+\alpha_{31}^\sc+\alpha_{32}^\sc+\alpha_{31}^\su+\alpha_{32}^\su)\\
&+h_2^2P(\alpha_{21}^\sb+\alpha_{23}^\sb+\alpha_{21}^\sc+\alpha_{23}^\sc+\tilde{\alpha}_{23}^\sc+\alpha_{21}^\su+\alpha_{23}^\su)\\
&+h_1^2P(\alpha_{12}^\sb+\alpha_{13}^\sb+\alpha_{12}^\sc+\alpha_{13}^\sc+\tilde{\alpha}_{13}^\sc+\alpha_{13}^\su)+1\\
=&h_3^2P(2\alpha_{31}^\sb+2\alpha_{32}^\sb+2\alpha_{31}^\sc+2\alpha_{32}^\sc+\alpha_{31}^\su+\alpha_{32}^\su)\\
&+h_2^2P(2\alpha_{21}^\sb+2\alpha_{23}^\sc+2\alpha_{21}^\sc+\alpha_{21}^\su+\alpha_{23}^\su)\\
&+h_1^2P\alpha_{13}^\su+1\\
\end{align*}
which follows by using \eqref{A12bVsA21b}, \eqref{A13bVsA31b}, \eqref{A23bVsA32b}, \eqref{A12cVsA23c}, \eqref{A23cVsA31c}, \eqref{A13cVsA32c}, and \eqref{A13cVsA21c}. For readability, let
\begin{align}
\sigma^2&=1+(2\alpha_{32}^\sb+2\alpha_{31}^\sb+2\alpha_{31}^\sc+2\alpha_{32}^\sc+\alpha_{32}^\su+\alpha_{31}^\su)h_3^2P.
\end{align}

Then, reliable decoding of $u_{12}^n$ is possible under the rate constraint given in equation \eqref{R12uU} at the top of page \pageref{R12uU}.
\begin{figure*}
\begin{align}
\label{R12uU}
R_{12}^\su&\leq C\left(\frac{\alpha_{12}^\su h_1^2P}{\sigma^2+2(\alpha_{21}^\sb+\alpha_{23}^\sc+\alpha_{21}^\sc)h_2^2P+\alpha_{13}^\su h_1^2P+\alpha_{21}^\su h_2^2P+\alpha_{23}^\su h_2^2P}\right)\\
\label{R13uU}
R_{13}^\su&\leq C\left(\frac{\alpha_{13}^\su h_1^2P}{\sigma^2+2(\alpha_{21}^\sb+\alpha_{23}^\sc+\alpha_{21}^\sc)h_2^2P+\alpha_{21}^\su h_2^2P+\alpha_{23}^\su h_2^2P}\right)\\
\label{R21uU}
R_{21}^\su&\leq C\left(\frac{\alpha_{21}^\su h_2^2P}{\sigma^2+2(\alpha_{21}^\sb+\alpha_{23}^\sc+\alpha_{21}^\sc)h_2^2P+\alpha_{23}^\su h_2^2P}\right)\\
\label{R23uU}
R_{23}^\su&\leq C\left(\frac{\alpha_{23}^\su h_2^2P}{\sigma^2+2(\alpha_{21}^\sb+\alpha_{23}^\sc+\alpha_{21}^\sc)h_2^2P}\right).
\end{align}
\hrule
\end{figure*}
After decoding $u_{12}^n$ ($m_{12}^\su$), its contribution is subtracted from the received signal at the relay. The other signals $u_{13}^n$, $u_{21}^n$, and $u_{23}^n$ can be decoded reliably if equations \eqref{R13uU}, \eqref{R21uU}, and \eqref{R23uU} at the top of page \pageref{R13uU} are satisfied.

\begin{remark}
Notice that the relay decodes signals whose power is multiplied by the channel gain $h_1^2$ first (here $u_{12}^n$ and $u_{13}^n$) and then those whose power is multiplied by the channel gain $h_2^2$. This order is also followed next, where signals whose power is multiplied by the channel gain $h_3^2$ are decoded last. If more signals have the same coefficient, then the uni-directional streams are decoded first, followed by the cyclic ones, and finally the bi-directional ones.
\end{remark}

We now proceed with decoding the superposition of the cyclic communication signals $c_{21}^n$ and $\tilde{c}_{13}^n$. Due to the group structure of lattice codes and lattice alignment, $(h_2\lambda_{21}^\sc+h_1\tilde{\lambda}_{13}^\sc)\bmod h_2\Lambda_{21,c}^\sc$ is a point in the nested lattice code $(h_2\Lambda_{21}^\sc,h_2\Lambda_{21,c}^\sc)$. The relay decodes the superposition $(h_2\lambda_{21}^\sc+h_1\tilde{\lambda}_{13}^\sc)\bmod h_2\Lambda_{21,c}^\sc$. From Lemma \ref{Lemma:AlignDecodeRate}, the decoding of this superposition is possible reliably as long as 

\begin{align}
\label{R132cU}
R_{132}^\sc\leq C\left(\frac{2\alpha_{21}^\sc h_2^2P}{\sigma^2+2(\alpha_{21}^\sb+\alpha_{23}^\sc)h_2^2P}\right)-1/2.
\end{align}

Next, the superposition of the cyclic streams $c_{12}^n$ and $c_{23}^n$ is considered. The lattice point $(h_1\lambda_{12}^\sc+h_2\lambda_{23}^\sc)\bmod h_2\Lambda_{23,c}^\sc$, which is a point in the nested lattice code $(h_2\Lambda_{23}^\sc,h_2\Lambda_{23,c}^\sc)$, is decoded, leading to the rate constraint
\begin{align}
\label{R123cU}
R_{123}^\sc\leq C\left(\frac{2\alpha_{23}^\sc h_2^2P}{\sigma^2+2\alpha_{21}^\sb h_2^2P}\right)-1/2.
\end{align}

The superposition of the aligned bi-directional streams $b_{12}^n$ and $b_{21}^n$ is decoded afterwards, i.e., $(h_1\lambda_{12}^\sb+h_2\lambda_{21}^\sb)\bmod h_2\Lambda_{21,c}^\sb$, with a rate constraint for reliable decoding given by
\begin{align}
\label{R12bU}
R_{12}^\sb=R_{21}^\sb\leq C\left(\frac{2\alpha_{21}^\sb h_2^2P}{\sigma^2}\right)-1/2.
\end{align}

Now, we have finished decoding all signals whose power is multiplied by the channel gain $h_2^2$. We proceed with decoding those whose power is multiplied by the channel gain $h_3^2$. We start with the uni-directional communication signals $u_{31}^n$ and $u_{32}^n$ first, then the superposition of the cyclic communication signals $(h_1\lambda_{13}^\sc+h_3\lambda_{32}^\sc)\bmod h_3\Lambda_{32,c}^\sc$ and $(h_2\tilde{\lambda}_{23}^\sc+h_3\lambda_{31}^\sc)\bmod h_3\Lambda_{31,c}^\sc$, and then the superposition of the bi-directional signals $(h_1\lambda_{13}^\sb+h_3\lambda_{31}^\sb)\bmod h_3\Lambda_{31,c}^\sb$ are decoded in the given order, resulting in the following rate constraints
\begin{align}
\label{R31uU}
R_{31}^\su&\leq C\left(\frac{\alpha_{31}^\su h_3^2P}{1+(2\alpha_{32}^\sb+2\alpha_{31}^\sb+2\alpha_{31}^\sc+2\alpha_{32}^\sc+\alpha_{32}^\su)h_3^2P}\right)\\
\label{R32uU}
R_{32}^\su&\leq C\left(\frac{\alpha_{32}^\su h_3^2P}{1+2(\alpha_{32}^\sb+\alpha_{31}^\sb+\alpha_{31}^\sc+\alpha_{32}^\sc)h_3^2P}\right)\\
\label{R132cU2}
R_{132}^\sc&\leq C\left(\frac{2\alpha_{32}^\sc h_3^2P}{1+2(\alpha_{32}^\sb+\alpha_{31}^\sb+\alpha_{31}^\sc)h_3^2P}\right)-1/2\\
\label{R123cU2}
R_{123}^\sc&\leq C\left(\frac{2\alpha_{31}^\sc h_3^2P}{1+2(\alpha_{32}^\sb+\alpha_{31}^\sb)h_3^2P}\right)-1/2\\
\label{R31bU}
R_{31}^\sb&\leq C\left(\frac{2\alpha_{31}^\sb h_3^2P}{1+2\alpha_{32}^\sb h_3^2P}\right)-1/2.
\end{align}
Finally, the superposition of $b_{23}^n$ and $b_{32}^n$, i.e., $(h_2\lambda_{23}^\sb+h_3\lambda_{32}^\sb)\bmod h_3\Lambda_{32,c}^\sb$ is decoded, with reliable decoding being possible if the following rate constraint is fulfilled
\begin{align}
\label{R32bU}
R_{32}^\sb\leq C(2\alpha_{32}^\sb h_3^2P)-1/2.
\end{align}

\subsection{Downlink}
In the uplink phase, the relay has decoded all the uni-directional communication signals $u_{ij}^n$, the superposition of the cyclic signals 
\begin{align}
&(h_1\lambda_{12}^\sc+h_2\lambda_{23}^\sc)\bmod h_2\Lambda_{23,c}^\sc,\\
&(h_2\tilde{\lambda}_{23}^\sc+h_3\lambda_{31}^\sc)\bmod h_3\Lambda_{31,c}^\sc,\\ &(h_2\lambda_{21}^\sc+h_1\tilde{\lambda}_{13}^\sc)\bmod h_2\Lambda_{21,c}^\sc,\\ 
&(h_1\lambda_{13}^\sc+h_3\lambda_{32}^\sc)\bmod h_3\Lambda_{32,c}^\sc, 
\end{align}
and the superposition of the bi-directional communication signals 
\begin{align}
&(h_1\lambda_{12}^\sb+h_2\lambda_{21}^\sb)\bmod h_2\Lambda_{21,c}^\sb,\\
&(h_1\lambda_{13}^\sb+h_3\lambda_{31}^\sb)\bmod h_3\Lambda_{31,c}^\sb,\\
&(h_2\lambda_{23}^\sb+h_3\lambda_{32}^\sb)\bmod h_3\Lambda_{32,c}^\sb.
\end{align}

In the downlink, the relay encodes each of these decoded signals into a Gaussian codeword. Each uni-directional signal $u_{ij}^n$ is encoded into $t_{ij}^n$, the superposition of the cyclic signals is encoded as follows 
\begin{align}
(h_1\lambda_{12}^\sc+h_2\lambda_{23}^\sc)\bmod h_2\Lambda_{23,c}^\sc&\to s_{12}^n,\\
(h_2\tilde{\lambda}_{23}^\sc+h_3\lambda_{31}^\sc)\bmod h_3\Lambda_{31,c}^\sc&\to s_{31}^n,\\ (h_2\lambda_{21}^\sc+h_1\tilde{\lambda}_{13}^\sc)\bmod h_2\Lambda_{21,c}^\sc&\to s_{21}^n,\\ 
(h_1\lambda_{13}^\sc+h_3\lambda_{32}^\sc)\bmod h_3\Lambda_{32,c}^\sc&\to s_{32}^n, 
\end{align}
and the bi-directional signals as follows
\begin{align}
(h_1\lambda_{12}^\sb+h_2\lambda_{21}^\sb)\bmod h_2\Lambda_{21,c}^\sb&\to r_{21}^n,\\
(h_1\lambda_{13}^\sb+h_3\lambda_{31}^\sb)\bmod h_3\Lambda_{31,c}^\sb&\to r_{31}^n,\\
(h_2\lambda_{23}^\sb+h_3\lambda_{32}^\sb)\bmod h_3\Lambda_{32,c}^\sb&\to r_{32}^n.
\end{align}

The relay allocates a power $\beta_{ij}^\su P$ for $t_{ij}^n$, i.e., $t_{ij}\sim\mathcal{N}(0,\beta_{ij}^\su P)$, it allocates $\beta_{ij}^\sc P$ for $s_{ij}^n$ and $\beta_{ij}^\sb P$ for $r_{ij}^n$. For the power constraint to be satisfied, it is required that the sum of all $\beta_{ij}^\su$, $\beta_{ij}^\sc$, and $\beta_{ij}^\sb$, given by $\beta_\Sigma$ fulfills
\begin{align}
\label{Cond:Beta}
\beta_\Sigma=\sum \beta_{ij}^\su+\sum \beta_{ij}^\sc+\sum \beta_{ij}^\sb\leq1. 
\end{align}
The relay then sends the superposition of all these codewords
\begin{align}
x_r^n=r_{21}^n+r_{31}^n+r_{32}^n+s_{12}^n+s_{31}^n+s_{21}^n+s_{32}^n+\sum_{i\neq j}t_{ij}^n.
\end{align}

Let us now illustrate how the decoding of the desired signals is done at each of the nodes $U_1$, $U_2$, and $U_3$.

\subsubsection{Decoding at $U_3$}
$U_3$ decodes the signals intended to it in this order: $t_{13}^n$, $t_{23}^n$, $s_{31}^n$, $s_{32}^n$, $r_{31}^n$, $r_{32}^n$, while treating the other signals as noise. The necessary rate constraints for reliable decoding are
\begin{align}
\label{R13uD}
R_{13}^\su&\leq C\left(\frac{\beta_{13}^\su h_3^2P}{1+(\beta_{23}^\su+\beta_{32}^\sc+\beta_{31}^\sc+\beta_{31}^\sb+\beta_{32}^\sb+\beta)h_3^2P}\right)\\
R_{23}^\su&\leq C\left(\frac{\beta_{23}^\su h_3^2P}{1+(\beta_{32}^\sc+\beta_{31}^\sc+\beta_{31}^\sb+\beta_{32}^\sb+\beta)h_3^2P}\right)\\
R_{132}^\sc&\leq C\left(\frac{\beta_{32}^\sc h_3^2P}{1+(\beta_{31}^\sc+\beta_{31}^\sb+\beta_{32}^\sb+\beta)h_3^2P}\right)\\
R_{123}^\sc&\leq C\left(\frac{\beta_{31}^\sc h_3^2P}{1+(\beta_{31}^\sb+\beta_{32}^\sb+\beta)h_3^2P}\right)\\
R_{31}^\sb&\leq C\left(\frac{\beta_{31}^\sb h_3^2P}{1+(\beta_{32}^\sb+\beta)h_3^2P}\right)\\
\label{R32bD}
R_{32}^\sb&\leq C\left(\frac{\beta_{32}^\sb h_3^2P}{1+\beta h_3^2P}\right),
\end{align}
where $\beta=\beta_{12}^\su+\beta_{32}^\su+\beta_{12}^\sc+\beta_{21}^\sc+\beta_{21}^\sb+\beta_{21}^\su+\beta_{31}^\su$.

By decoding $t_{13}^n$ and $t_{23}^n$, the third user can obtain $u_{13}^n$ and $u_{23}^n$ and hence its desired uni-directional messages $m_{13}^\su$ and $m_{23}^\su$. By decoding $s_{32}^n$, the third user can obtain the superposition $(h_1\lambda_{13}^\sc+h_3\lambda_{32}^\sc)\bmod h_3\Lambda_{32,c}^\sc$. Knowing $\lambda_{32}^\sc$, $U_3$ can extract $\lambda_{13}^\sc$ and hence obtain the desired cyclic communication message $m_{13}^\sc$ (cf. Lemma \ref{Lemma:AlignDecode}). Similarly, by decoding $s_{31}^n$, $m_{23}^\sc$ is recovered. Finally, by decoding $r_{31}^n$ and $r_{32}^n$, the superposition of $(h_1\lambda_{13}^\sb+h_3\lambda_{31}^\sb)\bmod h_3\Lambda_{31,c}^\sb$ and $(h_2\lambda_{23}^\sb+h_3\lambda_{32}^\sb)\bmod h_3\Lambda_{32,c}^\sb$ can be obtained, and consequently by using Lemma \ref{Lemma:AlignDecode}, the bi-directional messages $m_{13}^\sb$ and $m_{23}^\sb$. Thus, the third user obtains all its desired messages.

\subsubsection{Decoding at $U_2$}
Since the third user can decode its desired messages, the second user, having a stronger channel ($h_2^2\geq h_3^2$), can also decode all messages intended to user 3. After decoding the messages intended to user 3, the second user decodes its intended signals $t_{12}^n$, $t_{32}^n$, $s_{12}^n$, $s_{21}^n$, and $r_{21}^n$ successively in this order while treating the remaining signals as noise. The necessary rate constraints for reliable decoding are
\begin{align}
\label{R12uD}
R_{12}^\su&\leq C\left(\frac{\beta_{12}^\su h_2^2P}{1+(\beta_{32}^\su+\beta_{12}^\sc+\beta_{21}^\sc+\beta_{21}^\sb+\beta_{21}^\su+\beta_{31}^\su)h_2^2P}\right)\\
R_{32}^\su&\leq C\left(\frac{\beta_{32}^\su h_2^2P}{1+(\beta_{12}^\sc+\beta_{21}^\sc+\beta_{21}^\sb+\beta_{21}^\su+\beta_{31}^\su)h_2^2P}\right)\\
R_{123}^\sc&\leq C\left(\frac{\beta_{12}^\sc h_2^2P}{1+(\beta_{21}^\sc+\beta_{21}^\sb+\beta_{21}^\su+\beta_{31}^\su)h_2^2P}\right)\\
R_{132}^\sc&\leq C\left(\frac{\beta_{21}^\sc h_2^2P}{1+(\beta_{21}^\sb+\beta_{21}^\su+\beta_{31}^\su)h_2^2P}\right)\\
\label{R21bD}
R_{21}^\sb&\leq C\left(\frac{\beta_{21}^\sb h_2^2P}{1+(\beta_{21}^\su+\beta_{31}^\su)h_2^2P}\right).
\end{align}

By decoding $t_{12}^n$ and $t_{32}^n$, $U_2$ recovers $m_{12}^\su$ and $m_{32}^\su$. By decoding $s_{12}^n$, $U_2$ can obtain the superposition $(h_1\lambda_{12}^\sc+h_2\lambda_{23}^\sc)\bmod h_2\Lambda_{23,c}^\sc$. Since $U_2$ knows $\lambda_{23}^\sc$, the signal $\lambda_{12}^\sc$ can be extracted which recovers $m_{12}^\sc$. By decoding $s_{21}^n$, $U_2$ can obtain the superposition $(h_2\lambda_{21}^\sc+h_1\tilde{\lambda}_{13}^\sc)\bmod h_2\Lambda_{21,c}^\sc$. Since $\lambda_{21}^\sc$ is known by $U_2$, $\tilde{\lambda}_{13}^\sc$ can be extracted, and hence the message $m_{13}^\sc$. Notice that $m_{13}^\sc$ is not desired by $U_2$, but it can be used in combination with $s_{32}^n$ (recall that this can be decoded by $U_2$ since it can be decoded by $U_3$) to obtain $m_{32}^\sc$ which is a desired message. Finally, by decoding $r_{21}^n$ and $r_{32}^n$, $m_{12}^\sb$ and $m_{32}^\sb$ can be obtained. Consequently, all messages intended to $U_2$ are successfully decoded.

\subsubsection{Decoding at $U_1$}
The first user also decodes all signals that are decodable by $U_2$ and $U_3$. The cyclic messages $m_{21}^\sc$ and $m_{31}^\sc$ can be obtained from $s_{12}^n$, $s_{31}^n$, and $s_{21}^n$ in a manner similar to decoding the cyclic messages by $U_2$. The bi-directional messages $m_{21}^\sb$ and $m_{31}^\sb$ are also obtained similarly from $r_{21}^n$ and $r_{31}^n$. Then, it decodes the remaining signals $t_{21}^n$ and $t_{31}^n$ if the following rate constraints are satisfied
\begin{align}
\label{R21uD}
R_{21}^\su&\leq C\left(\frac{\beta_{21}^\su h_1^2P}{1+\beta_{31}^\su h_1^2P}\right)\\
\label{R31uD}
R_{31}^\su&\leq C\left(\beta_{31}^\su h_1^2P\right).
\end{align}

Thus, the uni-directional messages $m_{21}^\su$ and $m_{31}^\su$ are obtained from $t_{21}^n$ and $t_{31}^n$.  This recovers all messages intended to $U_1$.

Finally, from \eqref{R13uD}-\eqref{R31uD}, we can calculate the sum of all $\beta_{ij}^\su$, $\beta_{ij}^\sc$, and $\beta_{ij}^\sb$ at the relay $\beta_\Sigma$ to be as in equation \eqref{BetaSigma} at the top of page \pageref{BetaSigma}, which must be less than 1.
\begin{figure*}
\begin{align}
\label{BetaSigma}
\beta_\Sigma&=\frac{2^{2(R_{13}^\su+R_{23}^\su+R_{132}^\sc+R_{123}^\sc+R_{31}^\sb+R_{32}^\sb)}-1}{h_3^2P}+2^{2(R_{13}^\su+R_{23}^\su+R_{132}^\sc+R_{123}^\sc+R_{31}^\sb+R_{32}^\sb)}\frac{2^{2(R_{12}^\su+R_{32}^\su+R_{123}^\sc+R_{132}^\sc+R_{21}^\sb)}-1}{h_2^2P}\nonumber\\
&\quad+2^{2(R_{13}^\su+R_{23}^\su+2R_{132}^\sc+2R_{123}^\sc+R_{31}^\sb+R_{32}^\sb+R_{12}^\su+R_{32}^\su+R_{21}^\sb)}\frac{2^{2(R_{21}^\su+R_{31}^\su)}-1}{h_1^2P}
\end{align}
\hrule
\end{figure*}

Let the region achieved by this scheme, for a given power allocation satisfying the power constraints, be denoted $\mathcal{R}_g$, given by
\begin{align}
\mathcal{R}_g=\left\{\vec{R}\in\mathbb{R}_+^6| \text{ \eqref{R12uU}-\eqref{R32bU} and \eqref{R13uD}-\eqref{R31uD} are satisfied}\right\}. 
\end{align}
Then we have the following inner bound.
\begin{theorem}
The union over all possible power allocations satisfying the rate constraints \eqref{Cond:Alpha1}-\eqref{Cond:Alpha3}, and \eqref{Cond:Beta} of the region $\mathcal{R}_g$ is an inner bound on the capacity region $\mathcal{C}_g$ of the GYC
\begin{align}
\mathcal{C}_g\supseteq\underline{\mathcal{C}}_g=\bigcup_{
\substack{\alpha_{ij}^\su,\ \alpha_{ij}^\sc,\ \alpha_{ij}^\sb\text{ satisfying \eqref{Cond:Alpha1}-\eqref{Cond:Alpha3}}\\ 
\beta_{ij}^\su,\ \beta_{ij}^\sc,\ \beta_{ij}^\sb\text{ satisfying \eqref{Cond:Beta}}}} \mathcal{R}_g.
\end{align}
\end{theorem}

Next, we prove that this inner bound is within a constant gap, independent of the channel parameters, of the outer bound $\overline{\mathcal{C}}_g$.

\section{Constant Gap Characterization of $\mathcal{C}_g$}
\label{Sec:Gap}
The provided scheme achieves, within a constant gap, the outer bound $\overline{\mathcal{C}}_g$. Namely, we have the following corollary.

\begin{corollary}
For the given GYC, the region $\underline{\mathcal{C}}'_g$ given by 
\begin{align}
\label{R3132}
R_{31}+R_{32}&\leq C(h_3^2P)-2\\
\label{R1323}
R_{13}+R_{23}&\leq C(h_3^2P)-2\\
\label{R121332}
R_{12}+R_{13}+R_{32}&\leq C(h_2^2P+h_3^2P)-3\\
\label{R132312}
R_{13}+R_{23}+R_{12}&\leq C(h_2^2P+h_3^2P)-3\\
\label{R123132}
R_{12}+R_{31}+R_{32}&\leq C(h_1^2P+h_2^2P)-3\\
\label{R132321}
R_{13}+R_{23}+R_{21}&\leq C(h_1^2P+h_3^2P)-3\\
\label{R213123}
R_{21}+R_{31}+R_{23}&\leq C((h_2+h_3)^2P)-7/2\\
\label{R213132}
R_{21}+R_{31}+R_{32}&\leq C((h_2+h_3)^2P)-7/2,
\end{align}
is achievable.
\end{corollary}

Clearly, $\underline{\mathcal{C}}'_g$ is within a constant gap of at most 7/6 bits per stream, of the outer bound $\overline{\mathcal{C}}_g$. Thus, proving the achievability of $\underline{\mathcal{C}}'_g$ characterizes the capacity of the Y-channel within a constant gap. The remainder of this section is devoted for proving this result.

For this purpose, we need to show that any rate tuple in $\underline{\mathcal{C}}'_g$ is achievable. To simplify the analysis, we split the 6-dimensional space of $\vec{R}$ into 8 different sectors:
\begin{align}
\label{case:1}1)&\hspace{0.5cm}R_{12}\geq R_{21},\ R_{13}\geq R_{31},\ R_{23}\geq R_{32},\\
\label{case:2}2)&\hspace{0.5cm}R_{12}\geq R_{21},\ R_{13}\geq R_{31},\ R_{23}\leq R_{32},\\
\label{case:3}3)&\hspace{0.5cm}R_{12}\geq R_{21},\ R_{13}\leq R_{31},\ R_{23}\geq R_{32},\\
\label{case:4}4)&\hspace{0.5cm}R_{12}\geq R_{21},\ R_{13}\leq R_{31},\ R_{23}\leq R_{32},\\
\label{case:5}5)&\hspace{0.5cm}R_{12}\leq R_{21},\ R_{13}\geq R_{31},\ R_{23}\geq R_{32},\\
\label{case:6}6)&\hspace{0.5cm}R_{12}\leq R_{21},\ R_{13}\geq R_{31},\ R_{23}\leq R_{32},\\
\label{case:7}7)&\hspace{0.5cm}R_{12}\leq R_{21},\ R_{13}\leq R_{31},\ R_{23}\geq R_{32},\\
\label{case:8}8)&\hspace{0.5cm}R_{12}\leq R_{21},\ R_{13}\leq R_{31},\ R_{23}\leq R_{32}.
\end{align}

As we shall see, the third and the sixth cases are particularly important, since in these cases we need the cyclic communication strategy to achieve $\underline{\mathcal{C}}'_g$.  In the other cases, as we see next, the cyclic communications strategy is not necessary for a constant gap characterization of the capacity region.

In the following, the message $m_{ij}$ is split into its three parts $m_{ij}^\sb$, $m_{ij}^\sc$, and $m_{ij}^\su$ according to the following rates (similar to Sections \ref{BDC}, \ref{CC}, and \ref{AUC})
\begin{align}
\label{R21b}
R_{21}^\sb&=\min\{R_{12},R_{21}\}\\
R_{31}^\sb&=\min\{R_{13},R_{31}\}\\
R_{32}^\sb&=\min\{R_{32},R_{23}\}\\
R_{123}^\sc&=\min\{R_{12}-R_{21}^\sb,R_{23}-R_{32}^\sb,R_{31}-R_{31}^\sb\}\\
R_{132}^\sc&=\min\{R_{21}-R_{21}^\sb,R_{32}-R_{32}^\sb,R_{13}-R_{31}^\sb\}\\
\label{Riju}
R_{ij}^\su&=R_{ij}-R_{ij}^\sb-R_{ijk}^\sc,\quad i\neq j\neq k\neq i.
\end{align}
Recall \eqref{BiDRates} and \eqref{CycRates}. We fix the rates of the sub-messages as given above, and then, we show that there exists a valid power allocation that achieves these rates as long as $\vec{R}$ is in $\underline{\mathcal{C}}'_g$. Next, we only consider case 3) in \eqref{case:3}, the remaining cases are similar and are considered in Appendix \ref{GapProof}.

\subsection*{Case 3) $R_{12}\geq R_{21}$, $R_{13}\leq R_{31}$, $R_{23}\geq R_{32}$:}
This is one case where cyclic communication is necessary for achieving $\underline{\mathcal{C}}'_g$. Consider a rate tuple $\vec{R}$ in $\underline{\mathcal{C}}'_g$. In this case, we use \eqref{R21b}-\eqref{Riju} to write
\begin{align}
R_{12}^\sb&=R_{21},\quad R_{12}^\su=R_{12}-R_{21}-R_{123}^\sc\\
R_{13}^\sb&=R_{13},\quad R_{31}^\su=R_{31}-R_{13}-R_{123}^\sc\\
R_{23}^\sb&=R_{32},\quad R_{23}^\su=R_{23}-R_{32}-R_{123}^\sc\\
R_{123}^\sc&=\min\{R_{12}-R_{21},R_{23}-R_{32},R_{31}-R_{13}\}.
\end{align}

Using \eqref{R12uU}-\eqref{R32bU}, we have
\begin{align}
\alpha_{32}^\sb&=\frac{2^{2R_{32}+1}-1}{2h_3^2P}\\
\alpha_{31}^\sb&=\frac{2^{2R_{13}+1}-1}{2h_3^2P}2^{2R_{32}+1}\\
\alpha_{31}^\sc&=\frac{2^{2R_{123}^\sc+1}-1}{2h_3^2P}2^{2R_{32}+2R_{13}+2}\\
\alpha_{31}^\su&=\frac{2^{2(R_{31}-R_{13}-R_{123}^\sc)}-1}{h_3^2P}2^{2R_{123}^\sc+2R_{32}+2R_{13}+3}\\
\alpha_{21}^\sb&=\frac{2^{2R_{21}+1}-1}{2h_2^2P}2^{2R_{31}+2R_{32}+3}\\
\alpha_{23}^\sc&=\frac{2^{2R_{123}^\sc+1}-1}{2h_2^2P}2^{2R_{21}+2R_{31}+2R_{32}+4}\\
\alpha_{23}^\su&=\frac{2^{2(R_{23}-R_{32}-R_{123}^\sc)}-1}{h_2^2P}2^{2R_{21}+2R_{31}+2R_{123}^\sc+2R_{32}+5}\\
\alpha_{12}^\su&=\frac{2^{2(R_{12}-R_{21}-R_{123}^\sc)}-1}{h_1^2P}2^{2R_{23}+2R_{21}+2R_{31}+5}.
\end{align}

Now, we check if this power allocation is valid. Let us consider $U_3$, and check if the power allocation parameters above satisfy the power constraint. We add $\alpha_{32}^\sb$, $\alpha_{31}^\sb$, $\alpha_{31}^\sc$, and $\alpha_{31}^\su$ to obtain
\begin{align*}
&\alpha_{32}^\sb+\alpha_{31}^\sb+\alpha_{31}^\sc+\alpha_{31}^\su\\
&=\frac{2^{2(R_{31}+R_{32}+2)}-1}{2h_3^2P}\\
&\quad+\frac{2^{2(R_{123}^\sc+R_{32}+R_{13}+\frac{3}{2})}-2^{2(R_{123}^\sc+R_{32}+R_{13}+2)}}{2h_3^2P}
\\
&<\frac{2^{2(R_{31}+R_{32}+2)}-1}{2h_3^2P}\\
&\stackrel{\eqref{R3132}}{\leq} 1
\end{align*}
where the last step follows since $\vec{R}\in\underline{\mathcal{C}}'_g$. Thus, the power constraint is satisfied at $U_3$. Now, consider $U_2$. Similarly, we can show that as long as $\vec{R}\in\underline{\mathcal{C}}'_g$, then
\begin{align*}
\alpha_{23}^\sb+\tilde{\alpha}_{23}^\sc+\alpha_{21}^\sb+\alpha_{23}^\sc+\alpha_{23}^\su&<\frac{2^{2(R_{23}+R_{21}+R_{31}+3)}-1}{2h_2^2P}\\
&\stackrel{\eqref{R213123}}{\leq} 1.
\end{align*}
Thus, the power allocation also satisfies the power constraint at $U_2$. At $U_1$, we have
\begin{align*}
\alpha_{13}^\sb&+\alpha_{12}^\sb+\alpha_{12}^\sc+\alpha_{12}^\su\\
&\leq\frac{2^{2(R_{12}-R_{123}^\sc+R_{23}+R_{31}+3)}-1}{2h_1^2P}\\
&\leq\left\{
\begin{array}{l}
\frac{2^{2(R_{23}+R_{12}+R_{13}+3)}-1}{2h_1^2P},\quad \text{ if }R_{123}^\sc=R_{31}-R_{13}\\
\frac{2^{2(R_{23}+R_{21}+R_{31}+3)}-1}{2h_1^2P},\quad \text{ if }R_{123}^\sc=R_{12}-R_{21}\\
\frac{2^{2(R_{12}+R_{31}+R_{32}+3)}-1}{2h_1^2P},\quad \text{ if }R_{123}^\sc=R_{23}-R_{32}
\end{array}
\right.\\
&\leq 1
\end{align*}
which follows from \eqref{R132312}, \eqref{R123132}, and \eqref{R213123}. As a result, this power allocation is valid at all users. In the downlink, we calculate $\beta_\Sigma$ using \eqref{BetaSigma}
\begin{align*}
\beta_\Sigma=&\frac{2^{2(R_{23}+R_{13})}-1}{h_3^2P}+2^{2(R_{13}+R_{23})}\frac{2^{2R_{12}}-1}{h_2^2P}\\
&+2^{2(R_{13}+R_{23}+R_{12})}\frac{2^{2(R_{31}-R_{13}-R_{123}^\sc)}-1}{h_1^2P}\leq1\nonumber
\end{align*}
which holds due to \eqref{R1323}, \eqref{R132312}, \eqref{R123132}, and \eqref{R213123}. Since this power allocation is valid, $\underline{\mathcal{C}}'_g$ is achievable in this case. In Appendix \ref{GapProof}, we show that $\underline{\mathcal{C}}'_g$ is achievable in all the 8 sectors listed above (\eqref{case:1}-\eqref{case:8}). Therefore any rate tuple $\vec{R}$ which lies in $\underline{\mathcal{C}}'_g$ is achievable. We obtain the following theorem.

\begin{theorem}
The capacity region $\mathcal{C}_g$ of the GYC is within 7/6 bits per dimension of the outer bound $\overline{\mathcal{C}}_g$. In other words, if $\vec{R}\in\overline{\mathcal{C}}_g$, then $\vec{R}-(\frac{7}{6},\frac{7}{6},\frac{7}{6},\frac{7}{6},\frac{7}{6},\frac{7}{6})$ is achievable.
\end{theorem}
\begin{proof}
It can be easily verified that since $\vec{R}\in\overline{\mathcal{C}}_g$, then $\vec{R}'=\vec{R}-(\frac{7}{6},\frac{7}{6},\frac{7}{6},\frac{7}{6},\frac{7}{6},\frac{7}{6})$ is in $\underline{\mathcal{C}}'_g$. Thus $\vec{R}'$ is achievable by the provided scheme, which proves the statement of the theorem.
\end{proof}

\section{Discussion}
\label{Sec:Discussion}
Contrary to many bi-directional communications scenarios, where the cut-set bounds characterize the capacity of the setup in the linear-shift deterministic case \cite{AvestimehrKhajehnejadSezginHassibi,AvestimehrSezginTse}, we have shown that the cut-set bounds do not characterize the capacity of the DYC. It turns out that in such a multiway relaying setup, further bounds are required. Such bounds are derived based on a genie aided approach, leading to an outer bound on the capacity region of the DYC.

The achievability of this outer bound is established by using network coding ideas. The capacity achieving scheme is based on three different strategies, a bi-directional, a cyclic, and a uni-directional strategy. While the first and the last are used to establish the capacity of the bi-directional relay channel, the second is not. The nature of the DYC problem required the use of this cyclic strategy which takes care of bits communicated between the nodes in a cyclic manner, i.e., one bit from user 1 to user 2, one bit from user 2 to user 3, and one bit from user 3 to user 1, using the least amount of resources of the setup (levels). Showing the achievability of the outer bound, we established the capacity region of the DYC.

To extend this result to the Gaussian case, the GYC, a suitable approach is to use nested lattice codes to construct a scheme which mimics the capacity achieving scheme of the DYC. Owing to the group structure of lattice codes, the superposition of two properly designed codewords can be decoded, which mimics the decoding of the superposition (XOR) of bits at the relay in the DYC. By sending a superposition of codewords designed for bi-directional, cyclic, and uni-directional communication, designing a successive decoding/computation strategy and a forwarding strategy at the relay, and a successive decoding strategy at the users $U_1$, $U_2$ and $U_3$, we were able to characterize the capacity region of the GYC within a gap of 7/6 bits per dimension.

Note that this characterizes the sum-capacity of the GYC within a constant additive gap as well. Note that, in \cite{ChaabanSezginAvestimehr_YC_SC}, the sum-capacity of the GYC was characterized within a smaller gap, by restricting the analysis to the sum-capacity.

\begin{appendices}

\section{Uni-directional Communication over the DYC, Section \ref{AUC} Continued}
\label{App:UniDirectionalCases}
We only indicate the equations that are relevant for showing the sufficiency of the levels at the relay for each case. The analysis follows the same lines as in Section \ref{AUC}.

\subsubsection{Case $(R_{21}'',R_{23}'',R_{31}'')=(0,0,0)$}
Equations \eqref{TRB1}, \eqref{TRB5}, \eqref{CS3}, and \eqref{R1} imply 
\begin{align}
R_{32}''&\leq n_3'',\\
R_{13}''&\leq n_3'',\\
R_{12}''+R_{32}''+R_{13}''&\leq n_2''.
\end{align}

\subsubsection{Case $(R_{13}'',R_{12}'',R_{32}'')=(0,0,0)$}
Equations \eqref{TRB3}, \eqref{TRB4}, \eqref{CS3}, and \eqref{R2} imply 
\begin{align}
R_{31}''&\leq n_3'',\\
R_{23}''&\leq n_3'',\\
R_{21}''+R_{31}''+R_{23}''&\leq n_2''.
\end{align}

\subsubsection{Case $(R_{13}'',R_{21}'',R_{23}'')=(0,0,0)$}
Equations \eqref{TRB1}-\eqref{TRB5}, \eqref{CS3}, \eqref{R1}, and \eqref{R2} imply 
\begin{align}
R_{32}''+R_{31}''&\leq n_3'',\\
R_{32}''+R_{12}''&\leq n_2'',\\
R_{12}''+R_{32}''+R_{31}''&\leq n_1''.
\end{align}

\subsubsection{Case $(R_{12}'',R_{31}'',R_{32}'')=(0,0,0)$}
Equations \eqref{TRB1}, \eqref{TRB3}-\eqref{TRB5}, \eqref{CS3}, \eqref{R1}, and \eqref{R2} imply 
\begin{align}
R_{23}''+R_{13}''&\leq n_3'',\\
R_{21}''+R_{23}''&\leq n_2'',\\
R_{13}''+R_{21}''+R_{23}''&\leq n_1''.
\end{align}

\subsubsection{Case $(R_{12}'',R_{13}'',R_{23}'')=(0,0,0)$}
Equation \eqref{TRB3}, \eqref{TRB4}, \eqref{CS3}, and \eqref{R1} imply 
\begin{align}
R_{31}''+R_{32}''&\leq n_3'',\\
R_{21}''+R_{31}''+R_{32}''&\leq n_2''.
\end{align}

\section{Achievability of $\underline{\mathcal{C}}'_g$}
\label{GapProof}

\subsection{Case 1) $R_{12}\geq R_{21}$, $R_{13}\geq R_{31}$, $R_{23}\geq R_{32}$}
In this case, according to \eqref{R21b}-\eqref{Riju}, we have
\begin{align}
\label{RCase1-1}
R_{12}^\sb&=R_{21},\quad R_{12}^\su=R_{12}-R_{21}\\
\label{RCase1-2}
R_{13}^\sb&=R_{31},\quad R_{13}^\su=R_{13}-R_{31}\\
\label{RCase1-3}
R_{23}^\sb&=R_{32},\quad R_{23}^\su=R_{23}-R_{32}.
\end{align}
The rates of the other messages, i.e., $R_{123}^\sc$, $R_{132}^\sc$, $R_{21}^\su$, $R_{31}^\su$, and $R_{32}^\su$ are set to zero. In order to achieve these rates in the uplink, we substitute their values from \eqref{RCase1-1}-\eqref{RCase1-3} in \eqref{R12uU}-\eqref{R32bU} to obtain 
\begin{align}
\alpha_{32}^\sb&=\frac{2^{2R_{32}+1}-1}{2h_3^2P}\\
\alpha_{31}^\sb&=\frac{2^{2R_{31}+1}-1}{2h_3^2P}2^{2R_{32}+1}\\
\alpha_{21}^\sb&=\frac{2^{2R_{21}+1}-1}{2h_2^2P}2^{2R_{31}+2R_{32}+2}\\
\alpha_{23}^\su&=\frac{2^{2(R_{23}-R_{32})}-1}{h_2^2P}2^{2R_{21}+2R_{31}+2R_{32}+3}\\
\alpha_{13}^\su&=\frac{2^{2(R_{13}-R_{31})}-1}{h_1^2P}2^{2R_{21}+2R_{31}+2R_{23}+3}\\
\alpha_{12}^\su&=\frac{2^{2(R_{12}-R_{21})}-1}{h_1^2P}2^{2R_{21}+2R_{13}+2R_{23}+3}.
\end{align}
The power allocation parameters $\alpha_{23}^\sb$, $\alpha_{12}^\sb$, and $\alpha_{13}^\sb$ are calculated from $\alpha_{32}^\sb$, $\alpha_{21}^\sb$, and $\alpha_{31}^\sb$ by using \eqref{A12bVsA21b}, \eqref{A13bVsA31b}, and \eqref{A23bVsA32b}. All the remaining power allocation parameters in this case are equal to zero. Clearly, all $\alpha_{ij}^\sb,\alpha_{ij}^\su\geq0$. Now we need to show that they add up to a quantity less than 1 at each user, thus satisfying the power constraints. We notice that $\vec{R}\in\underline{\mathcal{C}}'_g$ implies
\begin{align*}
\alpha_{31}^\sb+\alpha_{32}^\sb&=\frac{2^{2R_{32}+2R_{31}+2}-1}{2h_3^2P}\stackrel{\eqref{R3132}}{\leq}1\\
\alpha_{21}^\sb+\alpha_{23}^\sb+\alpha_{23}^\su&\leq\frac{2^{2R_{23}+2R_{21}+2R_{31}+4}-1}{2h_2^2P}\stackrel{\eqref{R213123}}{\leq}1\\
\alpha_{12}^\sb+\alpha_{13}^\sb+\alpha_{12}^\su+\alpha_{13}^\su&\leq\frac{2^{2R_{12}+2R_{13}+2R_{23}+4}-1}{2h_1^2P}\stackrel{\eqref{R132312}}{\leq}1,
\end{align*}
Thus, this power allocation is valid in the uplink since it satisfies the power constraints at the sources. Now we check the achievability in the downlink. We use \eqref{BetaSigma} to calculate
\begin{align*}
\beta_\Sigma&=\frac{2^{2(R_{13}+R_{23})}-1}{h_3^2P}+2^{2(R_{13}+R_{23})}\frac{2^{2R_{12}}-1}{h_2^2P}\stackrel{\eqref{R1323},\ \eqref{R132312}}{\leq}1
\end{align*}
Thus, there exists a power allocation which satisfies the power constraint at the relay and achieves \eqref{RCase1-1}-\eqref{RCase1-3}, which proves the achievability of $\underline{\mathcal{C}}'_g$ in this case.

\subsection{Case 2) $R_{12}\geq R_{21}$, $R_{13}\geq R_{31}$, $R_{23}\leq R_{32}$}
Here, using \eqref{R21b}-\eqref{Riju} we set
\begin{align}
R_{12}^\sb&=R_{21},\quad R_{12}^\su=R_{12}-R_{21}\\
R_{13}^\sb&=R_{31},\quad R_{13}^\su=R_{13}-R_{31}\\
R_{23}^\sb&=R_{23},\quad R_{32}^\su=R_{32}-R_{23},
\end{align}
and we set the remaining rates, $R_{123}^\sc$, $R_{132}^\sc$, $R_{21}^\su$, $R_{23}^\su$, and $R_{31}^\su$, to zero. Using \eqref{R12uU}-\eqref{R32bU} we set 
\begin{align}
\alpha_{32}^\sb&=\frac{2^{2R_{23}+1}-1}{2h_3^2P}\\
\alpha_{31}^\sb&=\frac{2^{2R_{31}+1}-1}{2h_3^2P}2^{2R_{23}+1}\\
\alpha_{32}^\su&=\frac{2^{2(R_{32}-R_{23})}-1}{h_3^2P}2^{2R_{31}+2R_{23}+2}\\
\alpha_{21}^\sb&=\frac{2^{2R_{21}+1}-1}{2h_2^2P}2^{2R_{32}+2R_{31}+2}\\
\alpha_{13}^\su&=\frac{2^{2(R_{13}-R_{31})}-1}{h_1^2P}2^{2R_{21}+2R_{32}+2R_{31}+3}\\
\alpha_{12}^\su&=\frac{2^{2(R_{12}-R_{21})}-1}{h_1^2P}2^{2R_{21}+2R_{32}+2R_{13}+3}
\end{align}
to achieve $\vec{R}$. We calculate $\alpha_{23}^\sb$, $\alpha_{12}^\sb$, and $\alpha_{13}^\sb$ from $\alpha_{32}^\sb$, $\alpha_{21}^\sb$, and $\alpha_{31}^\sb$ by using \eqref{A12bVsA21b}, \eqref{A13bVsA31b}, and \eqref{A23bVsA32b}, and set the remaining power allocation parameters to zero. As long as $\vec{R}\in\underline{\mathcal{C}}'_g$, then
\begin{align*}
\alpha_{31}^\sb+\alpha_{32}^\sb+\alpha_{32}^\su&\leq\frac{2^{2R_{32}+2R_{31}+3}-1}{2h_3^2P}\stackrel{\eqref{R1323}}{\leq}1\\
\alpha_{21}^\sb+\alpha_{23}^\sb&\leq\frac{2^{2R_{32}+2R_{21}+2R_{31}+3}-1}{2h_2^2P}\stackrel{\eqref{R213132}}{\leq}1\\
\alpha_{12}^\sb+\alpha_{13}^\sb+\alpha_{12}^\su+\alpha_{13}^\su&\leq\frac{2^{2R_{12}+2R_{13}+2R_{32}+4}-1}{2h_1^2P}\stackrel{\eqref{R121332}}{\leq}1.
\end{align*}
Thus, this power allocation is valid since it satisfies the power constraints at the sources. In the downlink, we use \eqref{BetaSigma} to write
\begin{align*}
\beta_\Sigma&=\frac{2^{2(R_{13}+R_{23})}-1}{h_3^2P}+2^{2(R_{13}+R_{23})}\frac{2^{2(R_{12}+R_{32}-R_{23})}-1}{h_2^2P}\\
&\stackrel{\eqref{R1323},\ \eqref{R121332}}{\leq}1
\end{align*}
Thus $\underline{\mathcal{C}}'_g$ is also achievable in this case.

\subsection{Case 4) $R_{12}\geq R_{21}$, $R_{13}\leq R_{31}$, $R_{23}\leq R_{32}$}
Here, we have
\begin{align}
R_{12}^\sb&=R_{21},\quad R_{12}^\su=R_{12}-R_{21}\\
R_{13}^\sb&=R_{13},\quad R_{31}^\su=R_{31}-R_{13}\\
R_{23}^\sb&=R_{23},\quad R_{32}^\su=R_{32}-R_{23}.
\end{align}
The rates of the remaining messages, $R_{123}^\sc$, $R_{132}^\sc$, $R_{13}^\su$, $R_{21}^\su$, and $R_{23}^\su$, are set to zero. According to \eqref{R12uU}-\eqref{R32bU}, we set
\begin{align}
\alpha_{32}^\sb&=\frac{2^{2R_{23}+1}-1}{2h_3^2P}\\
\alpha_{31}^\sb&=\frac{2^{2R_{13}+1}-1}{2h_3^2P}2^{2R_{23}+1}\\
\alpha_{32}^\su&=\frac{2^{2(R_{32}-R_{23})}-1}{h_3^2P}2^{2R_{13}+2R_{23}+2}\\
\alpha_{31}^\su&=\frac{2^{2(R_{31}-R_{13})}-1}{h_3^2P}2^{2R_{32}+2R_{13}+2}\\
\alpha_{21}^\sb&=\frac{2^{2R_{21}+1}-1}{2h_2^2P}2^{2R_{32}+2R_{31}+2}\\
\alpha_{12}^\su&=\frac{2^{2(R_{12}-R_{21})}-1}{h_1^2P}2^{2R_{21}+2R_{32}+2R_{31}+3}.
\end{align}
The parameters $\alpha_{23}^\sb$, $\alpha_{12}^\sb$, and $\alpha_{13}^\sb$ are calculated from $\alpha_{32}^\sb$, $\alpha_{21}^\sb$, and $\alpha_{31}^\sb$ by using \eqref{A12bVsA21b}, \eqref{A13bVsA31b}, and \eqref{A23bVsA32b}. Now for $\vec{R}\in\underline{\mathcal{C}}'_g$, then
\begin{align*}
\alpha_{31}^\sb+\alpha_{32}^\sb+\alpha_{32}^\su+\alpha_{31}^\su&\leq\frac{2^{2R_{31}+2R_{32}+3}-1}{2h_3^2P}\stackrel{\eqref{R3132}}{\leq}1\\
\alpha_{21}^\sb+\alpha_{23}^\sb&\leq\frac{2^{2R_{32}+2R_{21}+2R_{31}+3}-1}{2h_2^2P}\stackrel{\eqref{R213132}}{\leq}1\\
\alpha_{12}^\sb+\alpha_{13}^\sb+\alpha_{12}^\su&\leq\frac{2^{2R_{12}+2R_{31}+2R_{32}+4}-1}{2h_1^2P}\stackrel{\eqref{R123132}}{\leq}1,
\end{align*}
Thus, this power allocation is valid since it satisfies the power constraints at the sources. In the downlink, we have
\begin{align*}
\beta_\Sigma=&\frac{2^{2(R_{13}+R_{23})}-1}{h_3^2P}+2^{2(R_{13}+R_{23})}\frac{2^{2(R_{12}+R_{32}-R_{23})}-1}{h_2^2P}\\
&+2^{2(R_{13}+R_{32}+R_{12})}\frac{2^{2(R_{31}-R_{13})}-1}{h_1^2P}\stackrel{\eqref{R1323},\ \eqref{R121332}, \eqref{R123132}}{\leq}1
\end{align*}
Since there exists a power allocation that is valid in both the uplink and the downlink, $\underline{\mathcal{C}}'_g$ is achievable in this case.

\subsection{Case 5) $R_{12}\leq R_{21}$, $R_{13}\geq R_{31}$, $R_{23}\geq R_{32}$}
In this case, we have
\begin{align}
R_{12}^\sb&=R_{12},\quad R_{21}^\su=R_{21}-R_{12}\\
R_{13}^\sb&=R_{31},\quad R_{13}^\su=R_{13}-R_{31}\\
R_{23}^\sb&=R_{32},\quad R_{23}^\su=R_{23}-R_{32},
\end{align}
and the remaining rates, $R_{123}^\sc$, $R_{132}^\sc$, $R_{12}^\su$, $R_{31}^\su$, and $R_{32}^\su$, equal to zero. We set
\begin{align}
\alpha_{32}^\sb&=\frac{2^{2R_{32}+1}-1}{2h_3^2P}\\
\alpha_{31}^\sb&=\frac{2^{2R_{31}+1}-1}{2h_3^2P}2^{2R_{32}+1}\\
\alpha_{21}^\sb&=\frac{2^{2R_{12}+1}-1}{2h_2^2P}2^{2R_{32}+2R_{31}+2}\\
\alpha_{23}^\su&=\frac{2^{2(R_{23}-R_{32})}-1}{h_2^2P}2^{2R_{12}+2R_{31}+2R_{32}+3}\\
\alpha_{21}^\su&=\frac{2^{2(R_{21}-R_{12})}-1}{h_2^2P}2^{2R_{12}+2R_{31}+2R_{23}+3}\\
\alpha_{13}^\su&=\frac{2^{2(R_{13}-R_{31})}-1}{h_1^2P}2^{2R_{21}+2R_{31}+2R_{23}+3}.
\end{align}
The parameters $\alpha_{23}^\sb$, $\alpha_{12}^\sb$, and $\alpha_{13}^\sb$ are calculated from $\alpha_{32}^\sb$, $\alpha_{21}^\sb$, and $\alpha_{31}^\sb$ by using \eqref{A12bVsA21b}, \eqref{A13bVsA31b}, and \eqref{A23bVsA32b}. Notice that $\vec{R}\in\underline{\mathcal{C}}'_g$ implies
\begin{align*}
\alpha_{31}^\sb+\alpha_{32}^\sb&=\frac{2^{2R_{31}+2R_{32}+2}-1}{2h_3^2P}\stackrel{\eqref{R3132}}{\leq}1\\
\alpha_{21}^\sb+\alpha_{23}^\sb+\alpha_{21}^\su+\alpha_{23}^\su&\leq\frac{2^{2R_{21}+2R_{31}+2R_{23}+4}-1}{2h_2^2P}\stackrel{\eqref{R213123}}{\leq}1\\
\alpha_{12}^\sb+\alpha_{13}^\sb+\alpha_{13}^\su&\leq\frac{2^{2R_{13}+2R_{21}+2R_{23}+4}-1}{2h_1^2P}\stackrel{\eqref{R132321}}{\leq}1,
\end{align*}
Thus, this power allocation is valid since it satisfies the power constraints at the sources. In the downlink, we calculate
\begin{align*}
\beta_\Sigma&=\frac{2^{2(R_{13}+R_{23})}-1}{h_3^2P}+2^{2(R_{13}+R_{23})}\frac{2^{2R_{12}}-1}{h_2^2P}\\
&\quad+2^{2(R_{13}+R_{23}+R_{12})}\frac{2^{2(R_{21}-R_{12})}-1}{h_1^2P}\stackrel{\eqref{R1323},\ \eqref{R132312},\ \eqref{R132321}}{\leq}1
\end{align*}
Thus, $\underline{\mathcal{C}}'_g$ is achievable in this case since there exists a power allocation that satisfies the power constraints and achieves any rate in $\underline{\mathcal{C}}'_g$.

\subsection{Case 6) $R_{12}\leq R_{21}$, $R_{13}\geq R_{31}$, $R_{23}\leq R_{32}$}
This is another case where cyclic communication is necessary for achieving $\underline{\mathcal{C}}$. Let 
\begin{align}
R_{12}^\sb&=R_{12},\quad R_{21}^\su=R_{21}-R_{12}-R_{132}^\sc\\
R_{13}^\sb&=R_{31},\quad R_{13}^\su=R_{13}-R_{31}-R_{132}^\sc\\
R_{23}^\sb&=R_{23},\quad R_{32}^\su=R_{32}-R_{23}-R_{132}^\sc\\
R_{132}^\sc&=\min\{R_{21}-R_{12},R_{32}-R_{23},R_{13}-R_{31}\}\\
\end{align}
and set the remaining rates, $R_{123}^\sc$, $R_{12}^\su$, $R_{23}^\su$, and $R_{31}^\su$, to zero. In order to achieve these rates, using \eqref{R12uU}-\eqref{R32bU} we set
\begin{align}
\alpha_{32}^\sb&=\frac{2^{2R_{23}+1}-1}{2h_3^2P}\\
\alpha_{31}^\sb&=\frac{2^{2R_{31}+1}-1}{2h_3^2P}2^{2R_{23}+1}\\
\alpha_{32}^\sc&=\frac{2^{2R_{132}^\sc+1}-1}{2h_3^2P}2^{2R_{23}+2R_{31}+2}\\
\alpha_{32}^\su&=\frac{2^{2(R_{32}-R_{23}-R_{132}^\sc)}-1}{h_3^2P}2^{2R_{132}^\sc+2R_{23}+2R_{31}+3}\\
\alpha_{21}^\sb&=\frac{2^{2R_{12}+1}-1}{2h_2^2P}2^{2R_{32}+2R_{31}+3}\\
\alpha_{21}^\sc&=\frac{2^{2R_{132}^\sc+1}-1}{2h_2^2P}2^{2R_{12}+2R_{32}+2R_{31}+4}\\
\alpha_{21}^\su&=\frac{2^{2(R_{21}-R_{12}-R_{132}^\sc)}-1}{h_2^2P}2^{2R_{12}+2R_{32}+2R_{132}^\sc+2R_{31}+5}\\
\alpha_{13}^\su&=\frac{2^{2(R_{13}-R_{31}-R_{132}^\sc)}-1}{h_1^2P}2^{2R_{21}+2R_{32}+2R_{31}+5}.
\end{align}
The parameters $\alpha_{23}^\sb$, $\alpha_{12}^\sb$, and $\alpha_{13}^\sb$ are calculated from $\alpha_{32}^\sb$, $\alpha_{21}^\sb$, and $\alpha_{31}^\sb$ by using \eqref{A12bVsA21b}, \eqref{A13bVsA31b}, and \eqref{A23bVsA32b}. The parameters $\alpha_{13}^\sc$ and $\tilde{\alpha}_{13}^\sc$ are calculated from $\alpha_{32}^\sc$ and $\alpha_{21}^\sc$ using \eqref{A13cVsA32c} and \eqref{A13cVsA21c}, respectively. The remaining power allocation parameters are set to zero. Now, we check if this power allocation is valid. At $U_2$ and $U_3$, we have
\begin{align*}
\alpha_{32}^\sb+\alpha_{31}^\sb+\alpha_{32}^\sc+\alpha_{32}^\su&\leq\frac{2^{2(R_{32}+R_{31}+2))}-1}{2h_3^2P}\stackrel{\eqref{R3132}}{\leq} 1\\
\alpha_{23}^\sb+\alpha_{21}^\sb+\alpha_{21}^\sc+\alpha_{21}^\su&\leq\frac{2^{2(R_{21}+R_{32}+R_{31}+3)}-1}{2h_2^2P}\stackrel{\eqref{R213132}}{\leq} 1,
\end{align*}
At $U_1$, we have
\begin{align*}
\alpha_{13}^\sb&+\alpha_{12}^\sb+\alpha_{13}^\sc+\tilde{\alpha}_{13}^\sc+\alpha_{13}^\su\\
&\leq\frac{2^{2(R_{13}+R_{21}+R_{32}-R_{132}^\sc+3)}-1}{2h_1^2P}\\
&\leq\left\{
\begin{array}{l}
\frac{2^{2(R_{13}+R_{21}+R_{23}+3)}-1}{2h_1^2P},\quad \text{ if } R_{132}^\sc=R_{32}-R_{23}\\
\frac{2^{2(R_{13}+R_{12}+R_{32}+3)}-1}{2h_1^2P},\quad \text{ if } R_{132}^\sc=R_{21}-R_{12}\\
\frac{2^{2(R_{21}+R_{32}+R_{31}+3)}-1}{2h_1^2P},\quad \text{ if } R_{132}^\sc=R_{13}-R_{31}
\end{array}
\right.\nonumber\\
&\leq 1
\end{align*}
which follows from \eqref{R121332}, \eqref{R132321}, and \eqref{R213132}. As a result, this power allocation is valid in the uplink in this case. In the downlink, we calculate $\beta_\Sigma$ from \eqref{BetaSigma}
\begin{align*}
\beta_\Sigma&=\frac{2^{2(R_{13}+R_{23})}-1}{h_3^2P}+2^{2(R_{13}+R_{23})}\frac{2^{2(R_{32}-R_{23}+R_{12})}-1}{h_2^2P}\\
&+2^{2(R_{13}+R_{32}+R_{12})}\frac{2^{2(R_{21}-R_{12}-R_{132}^\sc)}-1}{h_1^2P}\leq1
\end{align*}
from \eqref{R1323}, \eqref{R121332}, \eqref{R132321}, and \eqref{R213132}. There exists a power allocation which achieves these rates, which proves the achievability of $\underline{\mathcal{C}}'_g$ in this case.

\subsection{Case 7) $R_{12}\leq R_{21}$, $R_{13}\leq R_{31}$, $R_{23}\geq R_{32}$}

Here, we have
\begin{align}
R_{12}^\sb&=R_{12},\quad R_{21}^\su=R_{21}-R_{12}\\
R_{13}^\sb&=R_{13},\quad R_{31}^\su=R_{31}-R_{13}\\
R_{23}^\sb&=R_{32},\quad R_{23}^\su=R_{23}-R_{32}.
\end{align}
The remaining rates, $R_{123}^\sc$, $R_{132}^\sc$, $R_{12}^\su$, $R_{13}^\su$, and $R_{32}^\su$, are zero. We set
\begin{align}
\alpha_{32}^\sb&=\frac{2^{2R_{32}+1}-1}{2h_3^2P}\\
\alpha_{31}^\sb&=\frac{2^{2R_{13}+1}-1}{2h_3^2P}2^{2R_{32}+1}\\
\alpha_{31}^\su&=\frac{2^{2(R_{31}-R_{13})}-1}{h_3^2P}2^{2R_{13}+2R_{32}+2}\\
\alpha_{21}^\sb&=\frac{2^{2R_{12}+1}-1}{2h_2^2P}2^{2R_{32}+2R_{31}+2}\\
\alpha_{23}^\su&=\frac{2^{2(R_{23}-R_{32})}-1}{h_2^2P}2^{2R_{12}+2R_{31}+2R_{32}+3}\\
\alpha_{21}^\su&=\frac{2^{2(R_{21}-R_{12})}-1}{h_2^2P}2^{2R_{12}+2R_{31}+2R_{23}+3}.
\end{align}
The parameters $\alpha_{23}^\sb$, $\alpha_{12}^\sb$, and $\alpha_{13}^\sb$ are calculated from $\alpha_{32}^\sb$, $\alpha_{21}^\sb$, and $\alpha_{31}^\sb$ by using \eqref{A12bVsA21b}, \eqref{A13bVsA31b}, and \eqref{A23bVsA32b}. Since $\vec{R}\in\underline{\mathcal{C}}'_g$ then
\begin{align*}
\alpha_{31}^\sb+\alpha_{32}^\sb+\alpha_{31}^\su&\leq\frac{2^{2R_{31}+2R_{32}+3}-1}{2h_3^2P}\stackrel{\eqref{R3132}}{\leq}1,\\
\alpha_{21}^\sb+\alpha_{23}^\sb+\alpha_{21}^\su+\alpha_{23}^\su&\leq\frac{2^{2R_{21}+2R_{31}+2R_{23}+4}-1}{2h_2^2P}\stackrel{\eqref{R213123}}{\leq}1,\\
\alpha_{12}^\sb+\alpha_{13}^\sb&\leq\frac{2^{2R_{12}+2R_{32}+2R_{31}+3}-1}{2h_1^2P}\stackrel{\eqref{R123132}}{\leq}1.
\end{align*}
Thus, this power allocation is valid since it satisfies the power constraints at the sources. In the downlink, we calculate
\begin{align*}
\beta_\Sigma&=\frac{2^{2(R_{13}+R_{23})}-1}{h_3^2P}+2^{2(R_{13}+R_{23})}\frac{2^{2R_{12}}-1}{h_2^2P}\\
&\quad+2^{2(R_{13}+R_{23}+R_{12})}\frac{2^{2(R_{21}-R_{12}+R_{31}-R_{13})}-1}{h_1^2P}\\
&\stackrel{\eqref{R1323},\ \eqref{R132312},\ \eqref{R213123}}{\leq}1
\end{align*}
Thus $\underline{\mathcal{C}}'_g$ is also achievable in this case.

\subsection{Case 8) $R_{12}\leq R_{21}$, $R_{13}\leq R_{31}$, $R_{23}\leq R_{32}$}

In this case, we have
\begin{align}
R_{12}^\sb&=R_{12},\quad R_{21}^\su=R_{21}-R_{12}\\
R_{13}^\sb&=R_{13},\quad R_{31}^\su=R_{31}-R_{13}\\
R_{23}^\sb&=R_{23},\quad R_{32}^\su=R_{32}-R_{23}.
\end{align}
The remaining rates, $R_{123}^\sc$, $R_{132}^\sc$, $R_{12}^\su$, $R_{13}^\su$, and $R_{23}^\su$, are zero. We set
\begin{align}
\alpha_{32}^\sb&=\frac{2^{2R_{23}+1}-1}{2h_3^2P}\\
\alpha_{31}^\sb&=\frac{2^{2R_{13}+1}-1}{2h_3^2P}2^{2R_{23}+1}\\
\alpha_{32}^\su&=\frac{2^{2(R_{32}-R_{23})}-1}{h_3^2P}2^{2R_{13}+2R_{23}+2}\\
\alpha_{31}^\su&=\frac{2^{2(R_{31}-R_{13})}-1}{h_3^2P}2^{2R_{13}+2R_{32}+2}\\
\alpha_{21}^\sb&=\frac{2^{2R_{12}+1}-1}{2h_2^2P}2^{2R_{32}+2R_{31}+2}\\
\alpha_{21}^\su&=\frac{2^{2(R_{21}-R_{12})}-1}{h_2^2P}2^{2R_{12}+2R_{31}+2R_{32}+3}.
\end{align}
The parameters $\alpha_{23}^\sb$, $\alpha_{12}^\sb$, and $\alpha_{13}^\sb$ are calculated from $\alpha_{32}^\sb$, $\alpha_{21}^\sb$, and $\alpha_{31}^\sb$ by using \eqref{A12bVsA21b}, \eqref{A13bVsA31b}, and \eqref{A23bVsA32b}. As long as $\vec{R}\in\underline{\mathcal{C}}'_g$, then
\begin{align*}
\alpha_{31}^\sb+\alpha_{32}^\sb+\alpha_{31}^\su+\alpha_{32}^\su&\leq\frac{2^{2R_{31}+2R_{32}+3}-1}{2h_3^2P}\stackrel{\eqref{R3132}}{\leq}1\\
\alpha_{21}^\sb+\alpha_{23}^\sb+\alpha_{21}^\su&\leq\frac{2^{2R_{21}+2R_{31}+2R_{32}+4}-1}{2h_2^2P}\stackrel{\eqref{R213132}}{\leq}1\\
\alpha_{12}^\sb+\alpha_{13}^\sb&\leq\frac{2^{2R_{12}+2R_{32}+2R_{31}+3}-1}{2h_1^2P}\stackrel{\eqref{R123132}}{\leq}1.
\end{align*}
Thus, this power allocation is valid since it satisfies the power constraints at the sources. In the downlink, we calculate
\begin{align*}
\beta_\Sigma&=\frac{2^{2(R_{13}+R_{23})}-1}{h_3^2P}+2^{2(R_{13}+R_{23})}\frac{2^{2(R_{32}-R_{23}+R_{12})}-1}{h_2^2P}\\
&\quad+2^{2(R_{13}+R_{32}+R_{12})}\frac{2^{2(R_{21}-R_{12}+R_{31}-R_{13})}-1}{h_1^2P}\\
&\stackrel{\eqref{R1323},\ \eqref{R121332},\ \eqref{R213132}}{\leq}1
\end{align*}
Thus, $\underline{\mathcal{C}}$ is achievable in this case.

\end{appendices}

\bibliography{myBib}

\end{document}